\documentclass[reprint,superscriptaddress,amsmath,amssymb,aps,pra,]{revtex4-2}

\usepackage{graphicx,color}
\usepackage{dcolumn}
\usepackage{amsfonts,bm, bbm}
\usepackage{graphicx}
\usepackage{textcomp}
\usepackage{amsthm}
\usepackage{mathtools}
\usepackage[ruled,vlined,linesnumbered]{algorithm2e}
\usepackage{enumitem}
\usepackage{physics}
\usepackage[export]{adjustbox}
\usepackage[normalem]{ulem}
\usepackage{float}
\usepackage[nodisplayskipstretch]{setspace}
\usepackage[dvipsnames]{xcolor}

\mathtoolsset{showonlyrefs=true}
\theoremstyle{remark}
\newtheorem{definition}{Definition}[]
\newtheorem{theorem}{Theorem}[]

\newtheorem{corollary}{Corollary}[]

\newtheorem{lemma}{Lemma}[]
\newtheorem{remark}{Remark}[]

\begin{document}

\preprint{APS/123-QED}

\title{Limits of Fault-Tolerance on Resource-Constrained Quantum Circuits \\for Classical Problems}
\author{Uthirakalyani~G. \thanks{Equal contribution}}
\altaffiliation[Equal contribution.]{}
\affiliation{Department of Electrical Engineering, Indian Institute of Technology, Madras, Chennai 600036, India}

\author{Anuj K.\ Nayak \thanks{Equal contribution}}
\altaffiliation[Equal contribution.]{}
\affiliation{Coordinated Science Laboratory, University of Illinois at Urbana-Champaign, Urbana, IL 61801, USA
}
\thanks{This work was supported in part by National Science Foundation grant PHY-2112890.}

\author{Avhishek Chatterjee}
\affiliation{Department of Electrical Engineering, Indian Institute of Technology, Madras, Chennai 600036, India}

\author{Lav R.\ Varshney}
\affiliation{Coordinated Science Laboratory, University of Illinois at Urbana-Champaign, Urbana, IL 61801, USA
}
\date{\today}

\begin{abstract}

Existing lower bounds on redundancy in fault-tolerant quantum circuits are applicable when both the input and the intended output are quantum states. These bounds may not necessarily hold, however, when the intended output are classical bits, as in Shor's or Grover's algorithms. Here we show that indeed, noise thresholds obtained from existing bounds do not apply to a simple fault-tolerant implementation of the quantum phase estimation algorithm where the output quantum state suffers from noise before the measurement.  Then we obtain the first non-asymptotic lower bound on the minimum required redundancy for fault-tolerant quantum circuits with classical inputs and outputs.  Further, recent results show that due to physical resource constraints in quantum circuits, increasing redundancy  can increase noise, which in turn may render many fault-tolerance schemes useless. So it is of both practical and theoretical interest to characterize the effect of resource constraints on the fundamental limits of fault-tolerant quantum circuits. As an application of our lower bound, we characterize the fundamental limit of fault-tolerant quantum circuits with classical inputs and outputs under resource constraint-induced noise models. 

\end{abstract}

\maketitle

\section{Introduction}

Advantages of quantum computing over classical computing  \cite{Manin1980,Feynman1982}, especially when demonstrated mathematically \cite{deutsch1985quantum, deutsch1992rapid}, have spurred considerable interest. However, noise in quantum circuits heavily restricts the class of problems that can be solved using quantum hardware. Indeed, the formal term  \textit{NISQ} (Noisy Intermediate Scale Quantum) has been introduced to describe the current era where quantum processors are noise-limited \cite{Preskill2018}. 

To limit the corruption of quantum states due to noise,  pursuing fault-tolerant quantum circuits has led to a large literature in quantum error correction. Early papers demonstrated one can achieve arbitrary computational accuracy when physical noise is below a certain threshold. Achievability of any desired fault tolerance required a poly-logarithmic redundancy with respect to the size of the quantum circuit in these initial works  \cite{shor1996fault, steane1996error, aharonov1997fault, kitaev1997quantum}. More recent works extend such \emph{threshold theorems} to require only a  constant overhead \cite{gottesman2014fault, fawzi2020constant}, reminiscent of work in classical fault-tolerant computing \cite{Taylor1968, Varshney2011}. 

In this direction, there are fundamental lower bounds on redundancy for arbitrarily accurate computation \cite{fawzi2022lower, razborov2004upper, kempe2008upper, harrow2003robustness, UthirakalyaniNC2022}. However, all of these lower bounds are for quantum input/output, rather than classical input/output which is common for a large class of algorithms, such as those due to Deutsch-Jozsa \cite{deutsch1992rapid}, Shor \cite{Shor1994}, and Grover \cite{Grover1996}. The bounds in \cite{razborov2004upper,kempe2008upper} can be extended to classical input/output, but under restrictive assumptions on the final measurement. Here, we demonstrate by example that such lower bounds in quantum fault tolerance are not applicable for shallow quantum circuits with classical input/output, and prove a general alternate non-asymptotic bound.  As far as we know, this is the first non-asymptotic lower bound on fault tolerance for quantum circuits with classical input/output.  

The effects of noise on computational accuracy of quantum circuits are typically studied assuming the noise per physical qubit is constant with respect to the size of the circuit. Unfortunately, this is not true in many quantum devices today. Often due to limited physical resources such as energy \cite{ikonen2017energy}, volume \cite{monroe2013scaling}, or available bandwidth \cite{arute2019quantum}, they have physical noise levels that grow as the quantum computer grows \cite{fellous2021limitations}.  Fellous-Asiani, et al.~\cite{fellous2021limitations} introduce physical models of such scale-dependent noise and also aim to extend threshold theorems to this setting. However, the characterization of computational error (per logical qubit error)  is restricted to concatenated codes and does not apply to more general fault-tolerant schemes \cite{fawzi2020constant,gottesman2014fault}. As our lower bound is non-asymptotic, it is also applicable to scale-dependent noise. Using our lower bound and tools from optimization theory, we characterize the limits of scale dependence on fault-tolerant quantum circuits with classical input/output, applicable to any fault tolerance scheme.

The two motivations for the present work are therefore to obtain lower bounds on the required redundancy of a quantum circuit for computation with classical input/output, and to investigate the effect of resource constraints (like energy or volume) on this bound.  

The experimental finding that noise increases with more redundancy under resource constraints implies that simple per (logical) qubit redundancy cannot achieve arbitrary computational accuracy even if noise per physical qubit is below the fault-tolerance threshold, in contrast to conventional threshold theorems \cite{fellous2021limitations}. 
This is due to two opposing forces: improvement in accuracy due to increased redundancy and worse overall noise with redundancy due to scale dependence. In this regard, we find the sweet spot on redundancy for a desired computational accuracy using techniques from resource-limited (finite blocklength) quantum information theory.

The remainder of the paper is organized as follows. Section~\ref{sec:DJ_counter_example} gives a counterexample to illustrate the need for a new redundancy lower bound. Section~\ref{sec:model} gives mathematical models of computation, noise, and resource constraints that form the basis of our analysis. Then, the primary contributions follow. Section~\ref{sec:bounds} proves a non-asymptotic converse bound on redundancy required for classical computation on quantum circuits, drawing on one-shot capacity of classical-quantum channels (Theorem \ref{thm:finiteLB}).
 Section~\ref{sec:scale} analyzes the limits of scale-dependence for fault-tolerant computation, including closed-form and numerical solutions for some canonical noise models.
Finally, Section~\ref{sec:conclusion} concludes.

\section{Why is a New Bound Needed?}
\label{sec:DJ_counter_example}

In this section, we shall demonstrate the need for a new redundancy lower bound for quantum circuits with classical inputs and outputs with the help of a simple noisy computational model of a quantum phase estimation circuit. The quantum phase estimation algorithm is crucial to many important problems like discrete log and factoring. It is used to estimate the $n$-bit phase $\phi$ (with $\phi_1 \phi_2 ... \phi_n$ as the binary representation) of the eigenvalue $e^{j 2 \pi \phi}$ of a unitary operator $U_{\phi}$, given the corresponding eigenvector $\ket{u}$ and a controlled-$U_{\phi}$ \cite{nielsen2002quantum}. 

\begin{figure}[H]
    \centering    
    \includegraphics[scale=0.35]{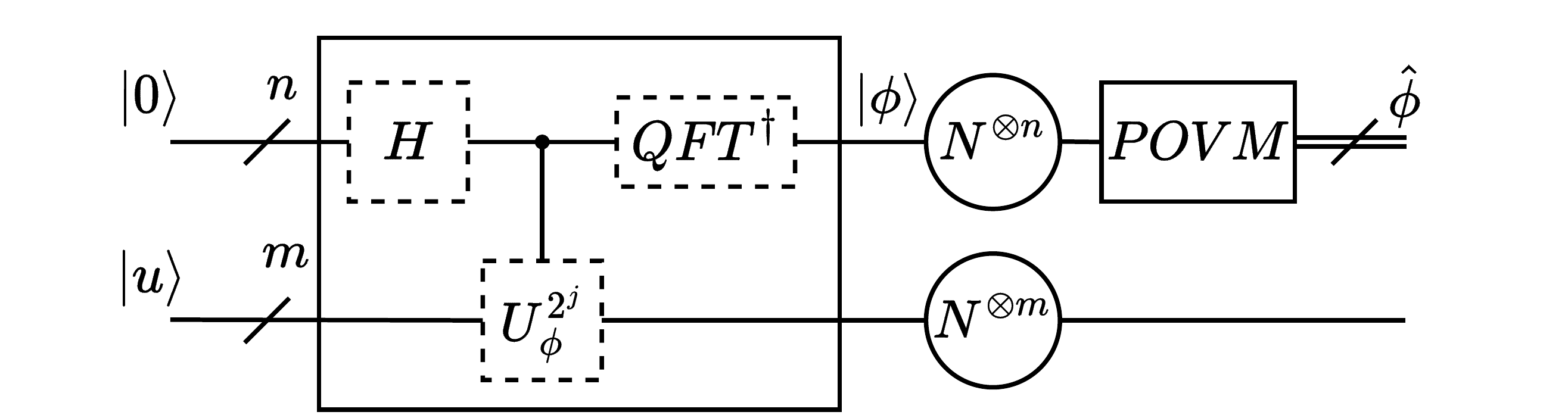}  
    \caption{Phase estimation algorithm in the presence of noise. 
    }
    \label{fig:PE_algo}
\end{figure}

Consider a simple {noisy computation model (in Figure~\ref{fig:PE_algo})}, where the qubits coming from a phase estimation circuit are corrupted independently by depolarizing (or erasure) noise before the final measurement (Please refer to Appendix. \ref{apndx:cptp_holevo_channels} for the definition of depolarizing and erasure channels). 

{Indeed, the noise model considered is a simplistic one, since it does not capture the noise in individual gates involved in quantum phase estimation. However, the existing bounds on quantum noise thresholds (e.g, \cite{fawzi2022lower}) for general purpose quantum computation are applicable to this model. In fact, the noisy computation model in Fig.~\ref{fig:PE_algo} conforms to the model in \cite[\S 1.1]{fawzi2022lower} with the entire quantum phase estimation circuit as a single layer quantum circuit followed by a layer of noise acting on each qubit independently. The sole purpose of this simplistic noise model and also of this section is to exhibit that the bounds for general purpose quantum computation are vacuous in the special but important case considered in this paper, i.e., classical inputs and outputs. We serve this purpose by showing that for a range of noise in this simplistic model where the existing bounds rule out any possibility of reasonably accurate computation, a simple fault-tolerant scheme can achieve arbitrarily high accuracy. Note that the fault-tolerant scheme proposed here for this simplistic noise model may not be useful in other settings.}

When the output of noiseless quantum phase estimation circuit $\ket{\phi}$ undergoes depolarizing noise, each qubit is replaced by a maximally mixed state with probability $p$, which results in each classical bit being flipped with probability $\frac{p}{2}$ independently upon measurement. In the case of erasure noise, each qubit (before measurement) flips to an orthogonal state $\ket{e}$ with probability $p$. Now, consider the following simple modification:
\begin{enumerate}
    \item Run quantum phase estimation algorithm, shown in Fig.~\ref{fig:PE_algo}, $T$ times. 
    \item Perform $T$ independent measurements on each {of the $n$ noisy output qubits $\mathcal{N}^{\otimes n}(\ketbra{\phi})$}. For erasure noise, declare the $i^{\text{th}}$ bit as $0$ (or $1$) if {the measurement outcomes for the corresponding qubit were $\hat{\phi_i} = 0$ (or $1$)} at least once in $T$ runs, and for depolarizing channel, declare the $i^{\text{th}}$ bit as $0$ (or $1$) if {the measurement outcomes for the corresponding qubit were $\hat{\phi_i} = 0$ (or $1$)} in more than $T/2$ runs.
\end{enumerate}

The above scheme can be seen as a fault-tolerant scheme {against the noise in Fig.~\ref{fig:PE_algo}} with redundancy of size $T$, i.e., {using a larger circuit with the algorithm repeated at most $T$ times. } 
We have the following guarantee on the performance of this simple fault-tolerant scheme for any $p<1$. 
\begin{theorem}
\label{thm:phEst}
    For a quantum phase estimation circuit that is corrupted by noise before the final measurement,
    \begin{enumerate}
        \item for depolarizing noise with probability $p$, if we choose $T \geq \frac{2 \ln(n/\epsilon)}{(1-p)^2}$, and 
        \item for erasure noise with probability $p$, if we choose $T \geq \left\vert \frac{\ln(\epsilon/n)}{\ln p}\right \vert$,
    \end{enumerate}
    then the modified circuit discussed above correctly outputs $\phi$ with probability at least $1-\epsilon$.
\end{theorem}

\begin{proof}
\textit{Depolarizing noise:}
When a qubit {$\ket{\phi_i} \in \{\ket{0}, \ket{1}\}$ }is subject to a depolarising noise with parameter $p$, then probability of error in detecting {$\phi_i$ (in a single run) is}

{
\begin{equation} 
    \mathbbm{P}\{ \phi_i \neq \hat{\phi_i}\} = Tr\left(\mathcal{N}(\ketbra{\phi_i}) \ketbra{\phi_i}\right) = \frac{p}{2},
\end{equation}
}
where $\hat{\phi_i}$ is the outcome of measurement in the computational basis, and $Tr(\cdot)$ is the trace operation. Now, the probability of error in estimating the $n-$bit phase (after $T$ runs){, denoted by $P_e$, is given by}
\begin{equation} \begin{split}
    P_e 
    &= \mathbbm{P}\{\cup_{i=1}^n\text{Error in $i^{th}$ bit}\}, \\
    &\leq n \, \mathbbm{P}\{\text{Error in one bit}\}. 
\end{split} \end{equation}
The last inequality is due to the symmetry of noise across qubits and union bound. Since we run the algorithm $T$ times and assign the majority to be the estimate,
{
\begin{equation} \begin{split}
    \mathbbm{P}\{\text{Error in one bit}\} &= \mathbbm{P}\{\text{$\phi_i \neq \hat{\phi}_i$ in $\geq T/2$ runs}\},
\end{split} \end{equation}
}

which is the tail of the binomial distribution w.p. $\frac{p}{2}$. Applying Hoeffding's inequality the probability of error is bounded as 
\begin{equation} \begin{split}
     \mathbbm{P}\{\text{Error in one bit}\} &\leq \exp(-2 \left(\frac{1-p}{2} \right)^2 T), \\
     P_e &\leq n \exp(-2 \left(\frac{1-p}{2} \right)^2 T) \leq \epsilon.
\end{split} \end{equation}
Choosing the number of runs to be at least
\begin{equation} \begin{split}
    T \geq \frac{|2 \ln{\frac{n}{\epsilon}}|}{(1-p)^2},
\end{split} \end{equation}
yields $P_e \leq \epsilon$ for any $p \in [0,1)$. 

\textit{Erasure noise:}
For erasure noise with probability of erasure $p$ and input $\ket{\phi_i} \in \{\ket{0}, \ket{1}\}$, the probability of error {in detecting $\phi_i$ is}
\begin{equation} 
    \mathbbm{P}\{ \phi_i \neq \hat{\phi_i}\} = Tr\left(\mathcal{N}(\ketbra{\phi_i}) \ketbra{\phi_i}\right)= p.
\end{equation}
The probability of error in estimating the n-bit phase (after $T$ runs) is bounded above as:
\begin{equation} \begin{split}
    P_e &= \mathbbm{P}\{\cup_{i=1}^n\text{Error in $i^{th}$ bit}\}\\
    &\leq n \mathbbm{P}\{\text{Error in one bit}\}
\end{split} \end{equation}
For erasure noise, error in estimating a bit occurs when the corresponding qubit is erased in all $T$ runs. Therefore,
{
\begin{equation} \begin{split}
    & \mathbbm{P}\{\text{Error in one bit}\} \\
    & = \mathbbm{P}\{\text{$\ket{\phi_i}$ flips to $\ket{e}$ in all $T$ runs,   $\ket{\phi_i} = \ket{1}$}\} \\
    & \quad + \mathbbm{P}\{\text{$\ket{\phi_i}$ flips to $\ket{e}$ in all $T$ runs,   $\ket{\phi_i} = \ket{0}$}\} \\
    &= \mathbbm{P}\{\text{$\ket{\phi_i}$ flips to $\ket{e}$ in all $T$ runs $\mid$  $\ket{\phi_i} = \ket{1}$}\} = p^T.
\end{split} \end{equation}
}
Therefore, the probability of error is bounded as
\begin{equation} 
     P_e \leq n p^T.
\end{equation}
Choosing the number of runs
\begin{equation} \begin{split}
    T \geq \left|\frac{\ln{\frac{\epsilon}{n}}}{\ln p}\right|,
\end{split} \end{equation}
we can achieve $P_e \leq \epsilon$ for any $p \in [0,1)$.
\end{proof}

The well-known threshold theorems \cite{harrow2003robustness,razborov2004upper, kempe2008upper, fawzi2022lower,  UthirakalyaniNC2022} imply that when the noise strength, $p$, is above a threshold, no fault-tolerant scheme with finite redundancy can compute a quantum state within reasonable accuracy. {This is because the lower bound on redundancy is given by $n/Q(\mathcal{N})$, where $n$ is the number of physical qubits and $Q(\mathcal{N})$ is the quantum capacity of channel $\mathcal{N}$ \cite{fawzi2022lower}}. For the shallow noisy quantum computational model for phase estimation discussed above, the best-known threshold for depolarizing \cite{kempe2008upper,fawzi2022lower,UthirakalyaniNC2022} and erasure \cite{fawzi2022lower,UthirakalyaniNC2022} noise are $\frac{1}{3}$ and $\frac{1}{2}$, respectively ({since quantum capacity vanishes above this threshold}). {However, through a noisy quantum phase estimation example, we have shown that when the inputs/outputs are classical with the performance criterion as the probability of error, it is indeed possible to achieve an arbitrarily small probability of error.} Thus, Theorem~\ref{thm:phEst} shows that the known redundancy lower bounds do not hold for quantum computation with classical input/output. This highlights the need for a bound which holds for classical inputs/outputs. A similar argument can be developed for the Deutsch-Jozsa algorithm and other well-known algorithms like discrete logarithm in the presence of noise.

Note that these do not imply the prior bounds on redundancy are incorrect; the apparent contradiction is due to differences in the definition of accuracy. Prior works use a notion of distance (or similarity) between the output quantum states of noiseless and noisy circuits to quantify accuracy. This requirement is too stringent when input and output are classical bits and error probability is a more suitable performance criterion  \cite{vonNeumann1956,Pippenger1988}. As such, we obtain a lower bound on the redundancy under the error probability criterion and then study the effect of resource constraints.

\section{Model}
\label{sec:model}
In this section, we discuss the model of computation, the relevant accuracy criteria, and the noise model.

\subsection{Model of Computation}

Consider the quantum circuit with classical inputs and classical outputs model in Fig.~\ref{fig:ModelCQC}, which is a standard model for gate-based quantum computation. This is denoted by $CQC: \{0,1\}^{n} \rightarrow \{0,1\}^{n}$ or equivalently $CQC(\mathbf{x})$ for $\mathbf{x} \in \{0,1\}^n$, where $n$ is the input size. The goal of the circuit is to realize a function $f:\{0,1\}^{n} \rightarrow \{0,1\}^{n}$.

\begin{figure}[H]
    \includegraphics[scale=0.26]{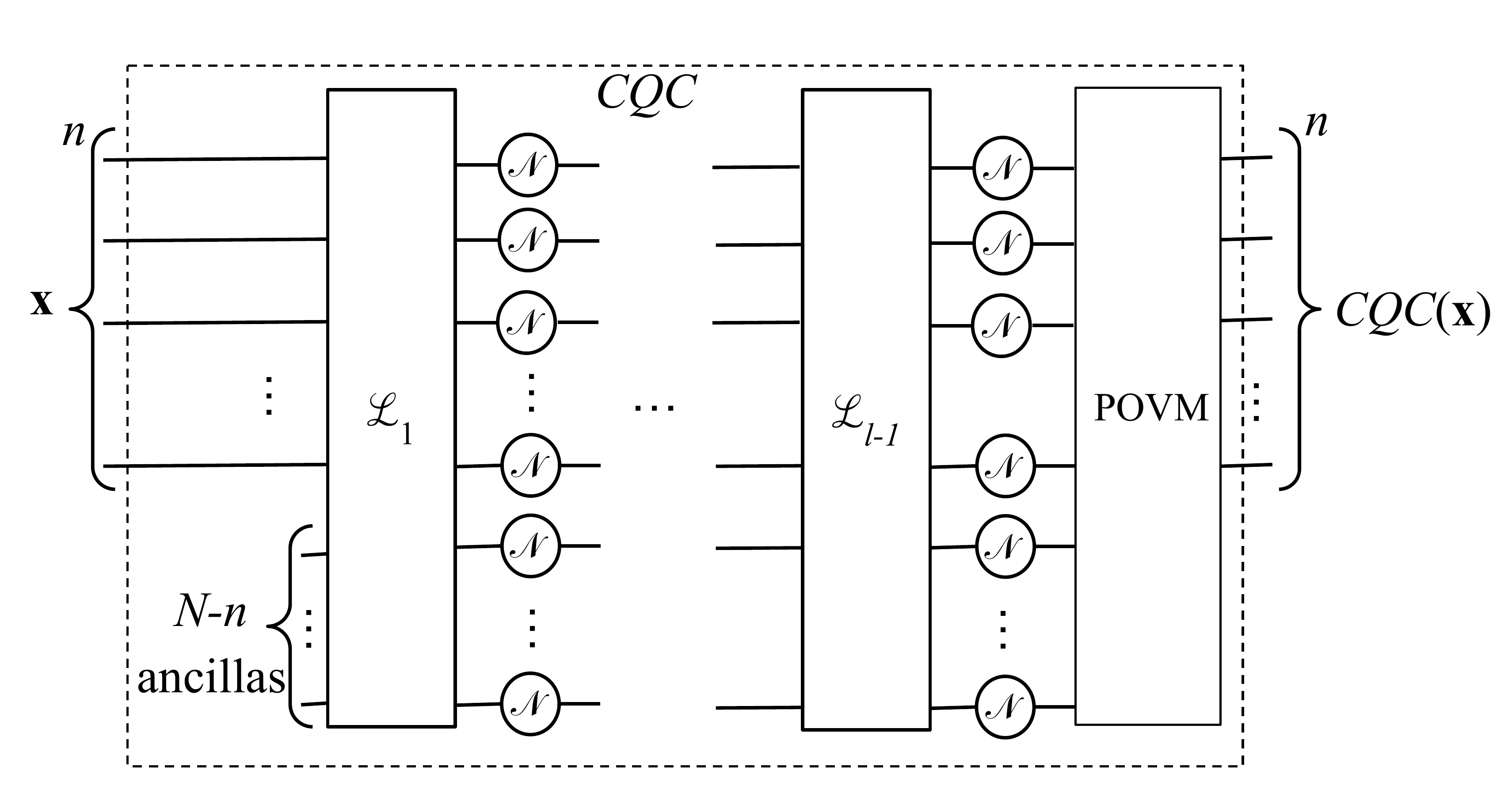}
    \caption{CQC model of computation: classical input, quantum computation, and classical output. }
    \label{fig:ModelCQC}
\end{figure}

The circuit consists of $l$ layers. The first layer takes $n$ classical inputs  ($\mathbf{x}$) as orthogonal quantum states $\ket{0}$ and $\ket{1}$ along with $N-n$ ancillas. It maps the input to a density operator of dimension $2^N$. Any subsequent layer $i$, for $2\le i \le l-1$, takes the output of the previous layer (layer $i-1$) as input. The output of any layer $i$, $1 \le i \le l-1$, is a  density operator of dimension $2^N$. The final layer (layer $l$) performs a POVM measurement and obtains classical output $CQC(\mathbf{x})$.

Each layer $i$, with $i \in \{1,2,\ldots,l-1\}$ is a noisy quantum operation. This is modeled as a noiseless quantum operation $\mathcal{L}_i$ on density operators of dimensions $2^N$ followed by $N$ i.i.d.\ quantum channels $\mathcal{N}$ (Fig.~\ref{fig:ModelCQC}). Finally, the last layer, layer $l$, performs a measurement (POVM), which yields a classical output. Thus the quantum circuit can be represented as a composition of quantum operations as $CQC(\mathbf{x}) = \mathcal{L}_l \circ \mathcal{N}^{\otimes N} \circ \mathcal{L}_{l-1} \circ \cdots \circ \mathcal{L}_2 \circ \mathcal{N}^{\otimes N} \circ \mathcal{L}_1 (\mathbf{x})$, where $\circ$ has the usual meaning of function composition. Next, we present the $\epsilon$-accuracy criteria for computations with classical I/O, taken from the seminal work by von Neumann \cite{vonNeumann1956}.

\begin{definition}[$\epsilon$-accuracy criteria]\label{defn:epsaccuracy}
Suppose $f(\cdot)$ is a classical function realized by a quantum circuit $CQC(\cdot)$ as defined in Sec.~\ref{sec:model}. Then the \emph{ $\epsilon$-accuracy} is:
\begin{equation}
    \mathbb{P}\{CQC(\mathbf{x}) \neq f(\mathbf{x})\} < \epsilon, \mbox{ for all } \mathbf{x} \in \{0,1\}^n.
    \label{cond:accuracy1}
\end{equation}
\end{definition}

In oracle-based algorithms like phase estimation and  Deutsch-Jozsa, there is no explicit input to the first layer, rather the classical parameters of the oracle are implicit classical inputs to the circuit. The results presented in this paper are directly applicable to that scenario as well.

\subsection{Noise Model}
Here, we consider only Holevo-additive channels characterized by a single parameter $p \in [0,1]$ and whose Holevo capacity is monotonically decreasing in $p$. The channels we specifically study are (i) $p$-erasure, (ii) $p$-depolarizing, and (iii) symmetric generalized amplitude damping channels, i.e., GADC$(p,\frac{1}{2})$.

In \cite{fellous2021limitations}, it was shown  that resource constraints can lead to an increase in  noise with increase in redundancy, \emph{scale-dependent noise}. A few models of scale-dependent noise, such as linear, polynomial, and exponential models,  have been studied in \cite{fellous2021limitations}. 
Let $k \triangleq N/n \geq 1$ be the redundancy and $p(k)$ be the noise strength when the redundancy is $k$. In the polynomial model, $p(k)=\min(p_0(1+\alpha (k-1))^{\gamma},1)$ and in the exponential model, $p(k)=\min(p_0\exp(\alpha (k-1)^{\gamma}),1)$ \cite{fellous2021limitations}. Here, $p_0\in[0,1]$ is the noise strength in the absence of any redundancy, i.e., $k=1$, and $\alpha$ and $\gamma$ are positive parameters. 

Intuitively, $p_0$ is the original noise strength of the particular quantum technology. As redundancy increases, more and more physical qubits have to share the same resource, which leads to increased interactions with the environment and among each other. These undesired interactions result in increased noise, which is captured by the models proposed in \cite{fellous2021limitations}.
In the interest of potentially wider applicability, we consider the following generic scale-dependent noise model, which includes the special cases discussed above.

\begin{definition}
\label{def:noise}
Noise $\mathcal{N}_p$ is parameterized by a single parameter $p \in [0,1]$ and the Holevo information $\chi(\mathcal{N}_p)$ is non-increasing in $p$. The parameter $p$ is a function of redundancy $k$, given by $p(k)\triangleq$ $\min(p(k;p_0,\bm{\theta}),1)$, where $\bm{\theta}$ is a tuple of non-negative parameters, and \\
~~~~(i) $p_0=p(1;p_0,\bm{\theta})$ for all $\bm{\theta}$, \\
~~~~(ii) for any $k\ge 1$, $p(k;p_0,\bm{\theta})$ is non-decreasing in any component of $\bm{\theta}$ and in $p_0$, given the other parameters are fixed.
\end{definition}

Here, $p_0$ represents the noise without redundancy, i.e., the initial noise without any resource constraint arising due to redundancy. 
Clearly, the polynomial and exponential models are special cases with $\bm{\theta}=(\alpha,\gamma)$.

The threshold for $p_0$, i.e., the minimum $p_0$ beyond which reliable quantum computation is not possible, was studied in \cite{fellous2021limitations} assuming concatenated codes for error correction. Here, we obtain a universal  threshold for all fault tolerance schemes. 

\section{Non-asymptotic Lower Bound on Redundancy}
\label{sec:bounds}
We obtain the lower bound by forming a mathematical relationship between the problem of $\epsilon$-accurate classical computation on a noisy quantum circuit and the problem of classical communication over a  finite number of i.i.d.\ uses of a quantum channel. Our approach builds on the following two simple observations.

First, a lower bound on redundancy obtained on a computational and noise model that have more capability would also be a lower bound for the original model. Second, a lower bound obtained under more relaxed accuracy criteria also applies to the original accuracy criteria.

Following the first observation, we obtain a redundancy bound for a model where there is no noise in layers $1, 2, \ldots, l-1$, and noise only at the last layer. 
Following the second observation, we use the following relaxation of the $\epsilon$-accuracy criteria.

Let the cardinality of the range of $f$ be $R_f$, and $\mathbf{x}^{(1)}, \mathbf{x}^{(2)}, \ldots, \mathbf{x}^{(R_f)} \in \{0,1\}^n$ be such that \ $|\{f(\textbf{x}^{(i)}): 1 \le i \le R_f\}|=R_f$. Then, we use the following accuracy criteria: for all $i=1, 2, \ldots, R_f$
\begin{equation}
    \mathbb{P}\{ CQC(\mathbf{x}^{(i)}) \neq f(\mathbf{x}^{(i)})\} < \epsilon.\label{cond:accuracy2_eqn}
\end{equation} 

\begin{figure}[H]
    \centering
    \includegraphics[scale=0.25]{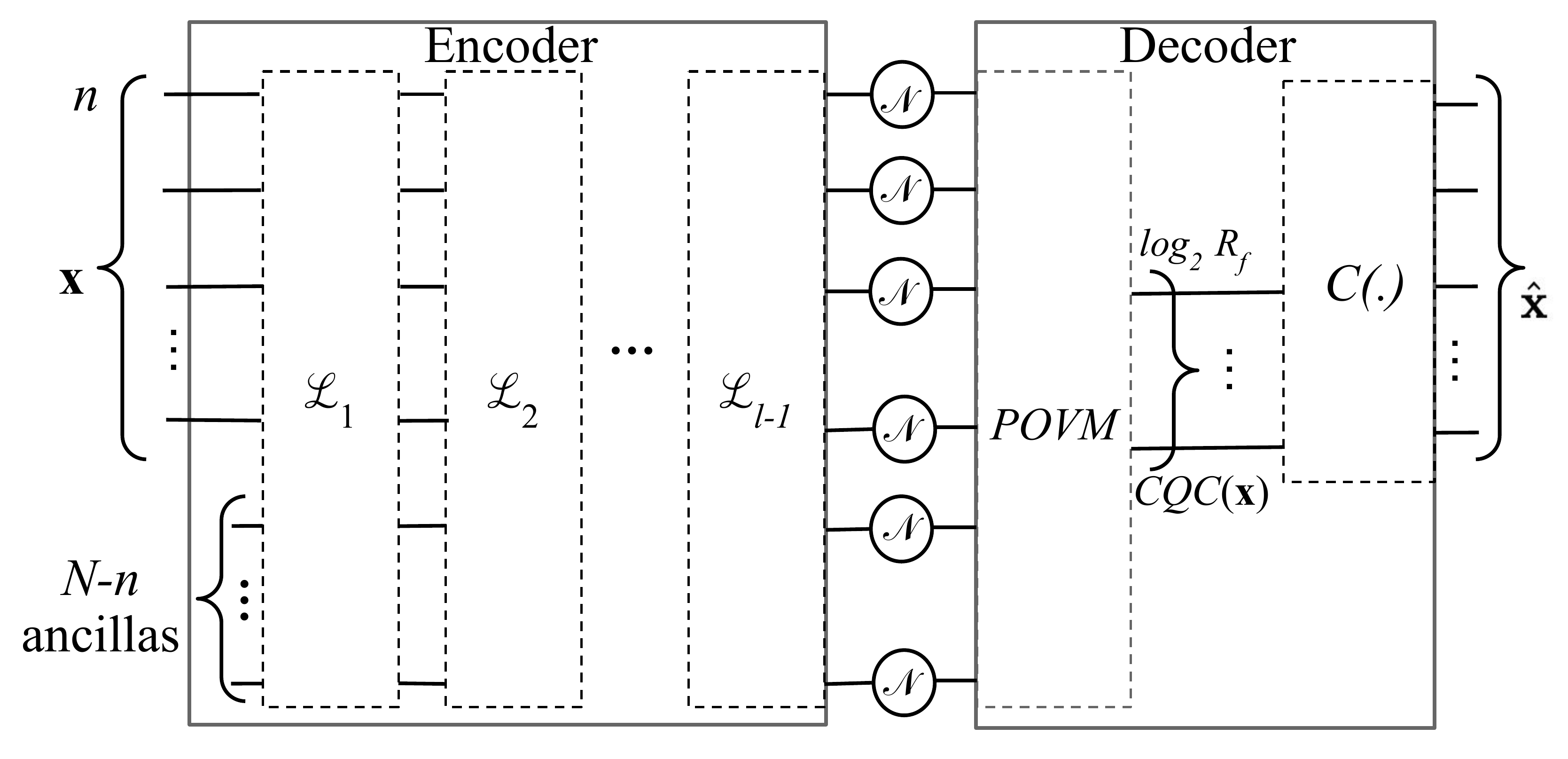}
    \caption{Reduction of noisy computation model in Fig. \ref{fig:ModelCQC} to a noisy communication problem. }
    \label{fig:ModelCQC_compute2comm}
\end{figure}

Clearly, this accuracy criteria is strictly weaker than the $\epsilon$-accuracy criteria when $f$ is not bijective. 
Further, suppose we have access to a hypothetically perfect (quantum or classical) circuit $C(\cdot)$ that can invert $f$ on the set $\{f(\textbf{x}^{(i)}): 1 \le i \le R_f\}$, where $f$ is bijective. Then, the relaxed accuracy criteria in \eqref{cond:accuracy2_eqn} becomes equivalent to
\begin{equation}
    \mathbb{P}\{ C(CQC(\mathbf{x}^{(i)})) \neq \mathbf{x}^{(i)}\} < \epsilon, \mbox{ for all } i = 1,2,\ldots,R_f\label{cond:accuracy3_eqn}
\end{equation} 
for the circuit shown in Figure~\ref{fig:ModelCQC_compute2comm}. This relaxed criterion can then be seen as the criterion for classical communication of $R_f$ messages with the maximum probability of error $\epsilon$ over a finite number of independent use of a channel.

Using results from one-shot classical communication over a quantum channel {\cite[\S 7.1.2]{khatri2020principles}} and using Holevo additiviy, we obtain the following lower bound on redundancy.

\begin{theorem}
\label{thm:finiteLB}
For $\epsilon$-accurate computation of any $n$-bit Boolean function $f$ using a quantum circuit in the presence of noise in Def.~\ref{def:noise}, the required redundancy is
\begin{equation}
    k > \frac{c(\epsilon, R_f, n)}{\chi(\mathcal{N}_{p(k)})}, \label{eqn:redundancy_bound_k}
\end{equation}
where,

$c(\epsilon, R_f, n) \triangleq (1-\epsilon) \frac{\log_2 R_f}{n} + \frac{\epsilon \log_2 \epsilon + (1-\epsilon) \log_2(1-\epsilon)}{n}$,

and $p(k)=\min(p(k;p_0,\bm{\theta}),1)$.
\end{theorem}

\begin{proof}
    For any additive quantum channel $\mathcal{N}$, an upper bound for classical communication over a quantum channel using an $(\mathcal{M}, N, \epsilon)$ code is  {\cite[\S 7.1.2]{khatri2020principles}}:
\begin{equation}
    \log_2(|\mathcal{M}|) \leq \frac{\chi(\mathcal{N}^{\otimes N}) + h_2(\epsilon)}{1-\epsilon},
\end{equation}
where {$\mathcal{M}$ is the message alphabet} and $h_2(.)$ is the binary entropy function. Assigning $|\mathcal{M}| = R_f$ yields
\begin{equation}
\begin{split}
    \epsilon > P_e &\geq 1-\frac{\chi(\mathcal{N}^{\otimes N})+h_2(P_e)}{\log_2 R_f}, \label{eqn:Pe_LB} \\
    & \geq 1-\frac{\chi(\mathcal{N}^{\otimes N})+h_2(\epsilon)}{\log_2 R_f}.    
\end{split}
\end{equation}
The last inequality holds, since $h_2(\cdot)$ is increasing in $[0,\frac{1}{2}]$.

Rearranging, we obtain
\begin{equation}
    \chi(\mathcal{N}^{\otimes N}) > (1-\epsilon) \log_2 R_f - h_2(\epsilon).
\end{equation}
Noting that Holevo information is additive, 
\begin{equation}
    N \chi(\mathcal{N}) >  (1-\epsilon) \log_2 R_f - h_2(\epsilon),
\end{equation}
\begin{equation}
    N > \frac{(1-\epsilon) \log_2 R_f - h_2(\epsilon)}{\chi(\mathcal{N})}. \label{eqn:redundancy_bound_N}
\end{equation}
Dividing both sides of the inequality by $n$, we obtain
\begin{equation}
    k\triangleq\frac{N}{n} > \frac{(1-\epsilon) \frac{\log_2 R_f}{n} - \frac{h_2(\epsilon)}{n}}{\chi(\mathcal{N}_{p(k)})}.
\end{equation}
\end{proof}

Unlike existing lower bounds on redundancy, this is a non-asymptotic bound and hence, is applicable to quantum computers of any size (including NISQ regime) and to any fault tolerance scheme. {A potential limitation is that our bound may be loose since we have reduced our noisy model to consist of only one layer of noise in our derivation.}

Note that the above bound is applicable for any Holevo-additive noise with a single parameter $p$, even when $p$ is a function of $k$. Thus, this bound can be used for understanding the limits of scale-dependent noise, as in the next section.

\section{Scale-dependent Noise: New Thresholds}
\label{sec:scale}

The terms in \eqref{eqn:redundancy_bound_k} can be rearranged to obtain $\epsilon$, a lower bound on the probability of error, over all fault-tolerant schemes, for different redundancy $k$. In Figure~\ref{fig:PeLB_vs_k}, the thin dashed lines plot the same for erasure noise with $p(k)=p_0$. This shows that there may exist a fault-tolerant scheme that can take the probability of error arbitrarily close to zero by increasing $k$.

\begin{figure}
    \centering
    \includegraphics[scale=0.49]{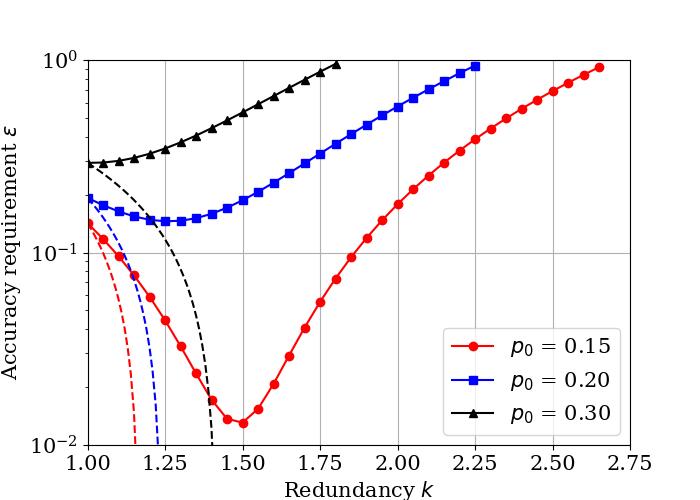}
    \caption{Solid and dashed lines are for erasure 
 noise with $p(k)=p_0$ and $p(k)=\min(p_0 (1+\alpha (k-1))^{\gamma}, 1)$, respectively.}
    \label{fig:PeLB_vs_k}
\end{figure}

On the contrary, the curves are significantly different for scale-dependent noise, i.e., when $p$ changes with $k$.
In Figure~\ref{fig:PeLB_vs_k}, we plot this for erasure $p(k) = \min(p_0 (1+\alpha (k-1))^{\gamma}, 1)$. First, observe that there is a certain value of $p_0$ ($p_0=0.3$ in this case) beyond which an increase in redundancy actually hurts performance. Second, observe that for each $p_0$ there is a minimum $\epsilon$ below which no fault-tolerant scheme can reach. 

These two observations lead to two interesting thresholds for initial noise $p_0$. The first threshold is the value of $p_0$ beyond which redundancy only hurts accuracy. 
The second threshold of interest is $p_0$ beyond which the $\epsilon$-accuracy criteria cannot be achieved using any fault tolerance scheme for a given $\epsilon>0$. Clearly, both these thresholds depend on the model of the scale-dependent noise and the parameters, e.g., $\alpha$, $\gamma$, therein. However, while the second threshold depends on $\epsilon$, the first threshold does not. 

The first threshold is similar to the threshold in \cite{fellous2021limitations}, but notably, their threshold calculation is limited to concatenated codes, whereas the threshold obtained here applies to any fault-tolerance scheme. The threshold is obtained using a derivative-based approach similar to \cite{fellous2021limitations}. For brevity, we omit details to focus on the second threshold, which has not been studied before.Please note that since both the noise thresholds are obtained using a genie-aided redundancy lower bound in Theorem~\ref{thm:finiteLB}, they are only (upper) bounds of the actual noise threshold, and the bounds could potentially be loose.

For a given $\bm{\theta}$ and $p(k;p_0,\bm{\theta})$, the second threshold $p_{th}(\bm{\theta}, \epsilon)$ is the minimum $p_0$ beyond which $\epsilon$-accurate computation is not possible. Thus, using Theorem~\ref{thm:finiteLB}, in particular, using \eqref{eqn:redundancy_bound_k}, $p_{th}(\bm{\theta}, \epsilon)$ is the minimum $p_0$ beyond which the criteria in Definition~\ref{defn:epsaccuracy} is not satisfied by any 
$k\ge 1$. Thus, $p_{th}(\bm{\theta}, \epsilon)$ is obtained by solving the following optimization problem.
\begin{equation}
\label{eqn:bilevel}
\begin{split}
\mbox{minimize }&  p_0 \\
\mbox{s.t. }& \min_{k\ge 1,~0\le p(k)\le 1}~ g_{\bm{\theta}}(k, p_0) ~\geq 0, 
\end{split}
\end{equation}
where $g_{\bm{\theta}}(k, p_0)=\frac{c(\epsilon, R_f, n)}{k} - \chi(\mathcal{N}_{p(k)})$.

Consider the following optimization problem
\begin{equation}
P_L: \min_{k\ge 1,~0\le p(k)\le 1}~ g_{\bm{\theta}}(k, p_0).  \nonumber
\end{equation}

Clearly, \eqref{eqn:bilevel} has the optimization problem $P_L$, which we refer to as the lower-level optimization problem, as a constraint. Thus, \eqref{eqn:bilevel} is a bi-level optimization problem. For a given set of $\bm{\theta}$ the solution to $P_L$ is a function of $p_0$, which we denote as $g^*_{\bm{\theta}}(p_0)$. Thus, the bi-level optimization problem in \eqref{eqn:bilevel} can also be written as
\begin{equation}
\begin{split}
& \mbox{min }  p_0   \\
& \mbox{s.t. } g^*_{\bm{\theta}}(p_0)\geq 0.   \label{eqn:bilevelAlt}
\end{split}
\end{equation}

In general, to compute the threshold $p_{th}(\bm{\theta}, \epsilon)$ one needs to solve  \eqref{eqn:bilevel}. However, for erasure noise and some special classes of $p(k;p_0,\bm{\theta})$, $g_{\bm{\theta}}(k,p_0)$ is convex in $k$ for all $p_0$. Hence,  the minimum of  the constraint in \eqref{eqn:bilevel} for a given $p_0$ can be derived in closed form. Then, solving for the smallest $p_0$ for which that minimum is non-negative gives the following theorem.

\begin{theorem}
\label{col:closed_erasure}
For erasure noise, thresholds for a fixed accuracy requirement $\epsilon>0$ are as follows:
\begin{enumerate}
   \item If $p(k;p_0,\alpha) = p_0 (1+\alpha (k-1))$, then 
    \begin{align*}p_{th}(\bm{\theta}, \epsilon) = 
        \begin{cases}
            1- c , & \mbox{ if } \alpha \geq \frac{c}{1-c}, \mbox{ and} \\
            \frac{\left( \sqrt{c\alpha} - \sqrt{c\alpha - \alpha +1} \right)^2}{(\alpha -1)^2}, & \mbox{ otherwise.}
        \end{cases}
    \end{align*}
    \item If $p(k;p_0,\gamma) = p_0 k^\gamma$, then 
    \begin{align*}p_{th}(\bm{\theta}, \epsilon) = 
        \begin{cases}
            1- c , & \mbox{ if }  \gamma \geq \frac{c}{1-c}, \mbox{ and} \\
            \frac{\left(\frac{\gamma}{c} \right)^{\gamma}}{\left(\gamma + 1 \right)^{\gamma + 1}}, & \mbox{otherwise}.
        \end{cases}
    \end{align*}
\end{enumerate}
Here $c=c(\epsilon,R_f,n)$ as defined in Theorem~\ref{thm:finiteLB}.

\begin{proof}
Consider the following procedure to find a \textit{closed-form} expression for $p_{th}(\bm{\theta}, \epsilon)$.
 \begin{enumerate}
    \item Minimize $g_{\bm{\theta}}(k, p_0)$ over $k$. Since $p(k;p_0, \bm{\theta})$ is non-decreasing in $k$, it is enough to minimize $g_{\bm{\theta}}(k, p_0)$ over $[1,k_{\max}]$, where $k_{\max} = \max\{k \mid p(k;p_0, \bm{\theta}) \leq 1\}$ . The minimum occurs at either $k=1$, $k=k_{\max}$ or a stationary point of $g_{\bm{\theta}}(k, p_0)$ in  $(1,k_{\max})$.
    \item Substitute the minimizer $k$ into $g_{\bm{\theta}}(k, p_0) \geq 0$, which yields an equation in $p_0, \bm{\theta}$.
    \item Solving the equation for $p_0$ yields a closed-form expression for $p_{th}(\bm{\theta}, \epsilon)$.
\end{enumerate}
{The derivation of $p_{th}(\bm{\theta}, \epsilon)$ for corresponding $p(k; p_0, \bm{\theta})$ is given in Appendix~\ref{sec:erasure_closed_form}.}
\end{proof}
\end{theorem}

Depending on the scale-dependent noise model, the optimization problem in the constraint of \eqref{eqn:bilevel} may or may not be convex.  In both cases, we develop algorithms that can solve the problem for all noise models in Definition~\ref{def:noise}.

For a general $p(k; p_0, \bm{\theta})$, a closed-form expression for $p_{th}(\bm{\theta}, \epsilon)$ in terms of $\bm{\theta}$ cannot be obtained, and therefore, $p_{th}(\bm{\theta}, \epsilon)$ must be computed numerically.

We develop Algorithm~\ref{algo:bisec_gd} to obtain $p_{th}(\bm{\theta}, \epsilon)$ by solving  bi-level optimization problem \eqref{eqn:bilevel}. In Algorithm~\ref{algo:bisec_gd}, we solve the alternate formulation \eqref{eqn:bilevelAlt} using the bisection method, while assuming  access to an oracle that computes $g^*_{\bm{\theta}}(p_0)$ for any $p_0$. Later, we also develop efficient algorithms that solve $P_L$ and obtain $g^*_{\bm{\theta}}(p_0)$ for any $p_0$. The proof of convergence of Algorithm~\ref{algo:bisec_gd} is given in Appendix~\ref{sec:proof_of_convg_bisec}.

Algorithm~\ref{algo:bisec_gd} computes the threshold $p_{th}(\bm{\theta}, \epsilon)$ (up to an error of $\delta_{p_0}$), for a pre-determined set of $\bm{\theta}$ (denoted by $\Theta$ of finite cardinality). Lines \ref{algo:bisec_start}--\ref{algo:bisec_end} describe the bisection method to compute $p_{th}(\bm{\theta}, \epsilon)$. Depending on whether $P_L$ is convex or non-convex, Algorithm~\ref{algo:proj_gd} or Algorithm~\ref{algo:line_search} is used to compute $g^*_{\bm{\theta}}(p_0)$, respectively.

\begin{algorithm}
    \caption{Algorithm to obtain $p_{th}(\bm{\theta}, \epsilon)$ or $\Bar{\Theta}_s$ numerically.}
    \label{algo:bisec_gd}
    \SetAlgoLined  
    Initialize the set $\Theta$, \textnormal{max\_iters}, $\delta_{p_0}$, $\delta$, $\Bar{\Theta}_s \leftarrow \{\}$, $k \leftarrow 1$\;    
    \BlankLine
    \For{\textnormal{each} $\bm{\theta} \in  \Theta$}{
        Initialize $i\leftarrow0$, $\Delta p_0 \leftarrow 1$, $p_0\leftarrow0.5$, $p^{-}_{0} \leftarrow 0$, $p^{+}_{0}\leftarrow1$\;
        \tcp{Bisection Method}
        \While{$\Delta p_0 > \delta_{p_0}$ \textnormal{and} $i<$\textnormal{max\_iters}}{ \label{algo:bisec_start}
            $p_{-1} \leftarrow p_0$\;
            $p_0 \leftarrow (p^{-}_{0} + p^{+}_{0})/2$\;
            $k_{\max} \leftarrow 1 + \alpha^{-1} ( (p_0)^{-\frac{1}{\gamma}} - 1 )$\;
            Solve $P_L$: $g^*_{\bm{\theta}}(p_0) = \min_{k \ge 1} g_{\bm{\theta}}(k,p_0)$\; \label{algo:PL}
            \If{$g^*_{\bm{\theta}}(p_0)>0$}{
                $p^{+}_{0} \leftarrow p_0$\;
            }
            \Else{
                $p^{-}_{0} \leftarrow p_0$\;
            }
            $\Delta p_0 \leftarrow |p_0 - p_{-1}|$\;
            $i \leftarrow i+1$\;
        } \label{algo:bisec_end}
        $p_{th}(\bm{\theta}, \epsilon) \leftarrow p_0$. \\
        $\Bar{\Theta}_s \leftarrow \Bar{\Theta}_s \bigcup \{(p_0, \alpha, \gamma)\}$\;
    }
\end{algorithm}

Obtaining $g^*_{\bm{\theta}}(\cdot)$ requires solving $P_L$. Next, we present efficient algorithms for solving $P_L$ for erasure, depolarizing, and symmetric GAD channels, and numerically obtain the converse surface for those noise models.

\subsection{Threshold for Erasure}
\label{sec:erasureConverseRegion}

In this section, we derive necessary conditions for $\epsilon$-accurate computation when the source of corruption of quantum states is erasure. Substituting for the classical capacity of QEC from \eqref{eqn:capacity_erasure} in \eqref{eqn:redundancy_bound_k} yields
\begin{equation}
    g_{\bm{\theta}}(k, p_0) = \frac{c(\epsilon, \eta, n)}{k} + p(k) - 1 < 0.
    \label{eqn:erasure_g}
\end{equation}

\begin{remark}
\label{remark:compact_k}
One can equivalently solve $P_L$ by restricting the range of $k$ to $[1,k_{\max}]$, where $k_{\max} = \max\{k \mid p(k;p_0, \bm{\theta}) \leq 1\}$. Also, $k_{\max}$ is finite and hence $[1,k_{\max}]$ is compact, which makes it convenient to solve \eqref{eqn:bilevel}. Therefore, one can replace line \ref{algo:PL} with $g^*_{\bm{\theta}}(p_0) = \min_{k \in [1,k_{\max}]} g_{\bm{\theta}}(k,p_0)$ to obtain the same value of threshold $p_{th}(\bm{\theta}, \epsilon)$. Please refer to Appendix~\ref{apndx:k_max} for more details.
\end{remark}

\subsubsection{Physical Noise $p(k; p_0, \bm{\theta})$ Convex in Redundancy $k$}

For the erasure channel, if $p(k; p_0, \bm{\theta})$ is convex, then $g_{\bm{\theta}}(k, p_0)$ is convex in $[1,k_{\max}]$, since the Holevo information $\chi(\mathcal{N}_{p})$ is affine in $p$. Therefore, the problem $P_L$ in \eqref{eqn:bilevel}, is convex, and from Remark~\ref{remark:compact_k}, the feasible set is compact.

A convex function over a compact set can be optimized using a gradient projection method given in \cite{wright1999numerical}. There are many algorithms to solve general gradient projection problems such as sequential quadratic programming (SQP) and  augmented Lagrangian methods that can be directly applied to solve $P_L$. Since, our problem is a one-dimensional convex problem (with only a Lipschitz gradient constraint) over a finite range $[1,k_{\max}]$, we provide a simple constant step-size gradient projection algorithm (Algorithm \ref{algo:proj_gd}). The algorithm starts from $k_{in} = 1$ when $g_{\bm{\theta}}(\cdot, p_0)$ is convex (w.l.o.g.).

\begin{algorithm}
\caption{Projected gradient descent routine.}
\label{algo:proj_gd}
\SetAlgoLined
\SetKwInOut{Input}{Input}
\SetKwInOut{Output}{Output}
\SetKwFunction{ProjGD}{ProjGD}

\SetKwProg{Fn}{Function}{:}{}

\Fn{\ProjGD{$g_{\bm{\theta}}$, $p_0$, $k_{in}$, $k_{\max}$, $L$, $\zeta$}}{
Initialize $j \leftarrow 0$, $k_j \leftarrow k_{in}$, $\xi = \frac{1}{L}$, $\Delta g = 2 \zeta$\;
\While{$\Delta g \geq \zeta$} {
        $d_j \leftarrow |g'_{\bm{\theta}}(k_j, p_0)|$\;
    
    $k_{j+1} \leftarrow \min \{ k_{\max}, \max\{1,k_j + \xi  d_j\}\}$\;
    $\Tilde{g} \leftarrow g_{\bm{\theta}}(k_{j+1},p_0)$\;
    $\Delta g \leftarrow |g_{\bm{\theta}}(k_{j},p_0) - \Tilde{g}|$\;
    $j \leftarrow j+1$\;
}
\Return{$\Tilde{g}, k_j$}\;
}
\end{algorithm}

Algorithm \ref{algo:proj_gd} solves $P_L$ optimally if step size ($\xi$) and stopping criterion ($\zeta$) are chosen appropriately. Sufficient conditions for convergence are: 1) $\xi \in (0, \frac{1}{L}]$, if $g'_{\bm{\theta}}(k,p_0) \triangleq \frac{\partial}{\partial k} g_{\bm{\theta}}(k,p_0)$ is $L$-Lipschitz over $[1, k_{\max}]$, and 2) stopping criterion provided in Definition~\ref{def:gd_stop_criterion1}. In all our computations, we choose $\xi = \frac{1}{L}$ as the step size for fast convergence.

\begin{definition}
\label{def:gd_stop_criterion1}
\textit{Stopping criterion 1:} Let $\{k_j\}$ be the iterates generated by the projected gradient descent algorithm (Algorithm~\ref{algo:proj_gd}), 
we use the following \textit{stopping criterion} for projected gradient descent algorithm:
\begin{equation}
    \label{eqn:stopping_criterion1}
    |g_{\bm{\theta}}(k_{j}, p_0) - g_{\bm{\theta}}(k_{j+1}, p_0)| < \frac{\delta^2}{2L k^2_{\max}} =: \zeta.
\end{equation}
\end{definition}

Then, it follows from the convexity and $L$-Lipschitz property of $g_{\bm{\theta}}(\cdot, p_0)$ that stopping criterion \eqref{eqn:stopping_criterion1} is a sufficient condition for convergence, which is $ g_{\bm{\theta}}(k_{j+1}, p_0) - g^*_{\bm{\theta}}(p_0)  \leq \delta$. The proof of convergence of Algorithm~\ref{algo:proj_gd} and derivation of Lipschitz constant for erasure noise are given in Appendix~\ref{sec:proof_of_convg_proj_gd}.

\subsubsection{Physical Noise $p(k;p_0, \bm{\theta})$ Non-convex in Redundancy $k$}

Suppose $p(k)$ is non-convex, then $g_{\bm{\theta}}(k;\bm{\theta}, \epsilon)$ is also non-convex. Hence, the lower-level problem $P_L$ cannot be solved using Algorithm~\ref{algo:proj_gd} (\texttt{ProjGD}). Therefore, we provide a line-search algorithm (Algorithm~\ref{algo:line_search}) to compute solution for a non-convex problem $P_L$. 

In Algorithm~\ref{algo:line_search}, the compact set $[1,k_{\max}]$ is traversed by successive gradient descent (or ascent) and perturbation over a one-dimensional non-convex function using an iterate starting from $k=1$ (w.l.o.g.) and moving in the positive $k$ direction. Lines \ref{algo:line_search_loop_begin}--\ref{algo:line_search_loop_end} include one iteration of Algorithm~\ref{algo:line_search}, which contains calls to \texttt{ProjGD} and \texttt{Perturb} as subroutines. The variable $\tilde{g}$ keeps track of the minimum value of $g_{\bm{\theta}}(\cdot, p_0)$ encountered thus far with an error of $\delta>0$.

\begin{algorithm}   
    \caption{Line search algorithm to find $\min_{k\geq1} g_{\bm{\theta}}(k,p_0)$, when $g_{\bm{\theta}}(k,p_0)$ is non-convex w.r.t. $k$.}
    \label{algo:line_search}
    \SetKwFunction{LineSearch}{LineSearch}
    \SetKwFunction{ProjGD}{ProjGD}
    \SetKwFunction{Perturb}{Perturb}
    
    \SetKwProg{Fn}{Function}{:}{}
    \Fn{\LineSearch{$g_{\bm{\theta}}$, $p_0$, $k_{\max}$, $L$, $\delta$}}{
    Initialize $i \leftarrow 0$, $k_i \leftarrow 1$;
    $\tilde{g} \leftarrow \underset{k \in \{1,k_{\max}\}}{\min} g_{\bm{\theta}}(k, p_0)$\; \label{algo:line_search_check_min_init}
    \While{$i<$\textnormal{max\_iters and} $k_i < k_{\max}$}{\label{algo:line_search_loop_begin}
    $g_i, k^-_i \leftarrow \texttt{ProjGD}(g_{\bm{\theta}}, p_0, k_i, k_{\max}, L, \delta)$\;
    $\hat{g}_i, \hat{k}_i, k_{i+1} \leftarrow \texttt{Perturb}(g_{\bm{\theta}}, p_0, k^-_i, k_{\max}, \delta)$\;
    $\tilde{g} \leftarrow \min \{ \tilde{g}, \hat{g}_i\}$\;
    $i \leftarrow i+1$;
    } \label{algo:line_search_loop_end}
    \Return $\tilde{g}$;
    }
    \Fn{\Perturb{$g_{\bm{\theta}}, p_0, k, k_{\max}, L, \delta$}}{
        $\Delta k \leftarrow \sqrt{\frac{2 \delta}{L}}$, $\xi \leftarrow \frac{1}{L}$, $\hat{g} \leftarrow g_{\bm{\theta}}(k, p_0)$, $\hat{k} \leftarrow k$\;
        $k' \leftarrow \min \{k_{\max}, k + \xi |g'_{\bm{\theta}}(k, p_0)| \}$\;
        \While{$|g_{\bm{\theta}}(k, p_0)-g_{\bm{\theta}}(k', p_0)| < \delta$ \textbf{and} $k < k_{\max}$}{
            $k \leftarrow \min\{k_{\max}, k + \Delta k\}$\;
            $k' \leftarrow \min \{k_{\max}, k + \xi |g'_{\bm{\theta}}(k, p_0)| \}$\; \label{algo:perturb_k_prime}
            $\hat{g} \leftarrow \min \{ \hat{g}, g_{\bm{\theta}}(k, p_0)\}$\;
            $\hat{k} \leftarrow \underset{k \in \{\hat{k},k\}}{\mbox{argmin}}\{\hat{g}, g_{\bm{\theta}}(k, p_0)\}$\;
        }
        \Return $\hat{g}, \hat{k}, k$\;
    }
\end{algorithm}

In Algorithm~\ref{algo:line_search} we reuse the \texttt{ProjGD} routine for gradient ascent/descent but with a different (more relaxed) stopping criterion than in Def.~\ref{def:gd_stop_criterion1}. 

\begin{definition}
\label{def:gd_stop_criterion2}
\textit{Stopping criterion 2:} Let $\{k_j\}$ be the iterates generated by the projected gradient descent algorithm (Algorithm~\ref{algo:proj_gd}). 
We use the following \textit{stopping criterion} for projected gradient descent algorithm:
\begin{equation}
    \label{eqn:stopping_criterion2}
    |g_{\bm{\theta}}(k_{j}, p_0) - g_{\bm{\theta}}(k_{j+1}, p_0)| < \delta =: \zeta.
\end{equation}
\end{definition}

\begin{definition}
\label{def:prtrb_stop_criterion}
\textit{Stopping criterion for} \texttt{Perturb} \textit{:}
Let $\{k_j\}$ be a sequence generated by \texttt{Perturb} routine.
\begin{equation}
    |g_{\bm{\theta}}(k_{j}, p_0) - g_{\bm{\theta}}(k'_{j}, p_0)| \geq \delta.
\end{equation}
where $k'_{j} = \min \{k_{\max}, k + \xi |g'_{\bm{\theta}}(k_j, p_0)| \}$ in line \ref{algo:perturb_k_prime} of Algorithm~\ref{algo:line_search}, and $g'_{\bm{\theta}}(z, p_0) = \frac{\partial g'_{\bm{\theta}}(k, p_0)}{\partial k}\Bigr|_{k = z}$.
\end{definition}

Note that the stopping criterion for \texttt{Perturb} is similar to Definition~\ref{def:gd_stop_criterion2}, but with the inequality reversed. Since the stopping criteria of \texttt{ProjGD} and \texttt{Perturb} are complementary, only one of the routines will be active during the execution of Algorithm~\ref{algo:line_search}.

The proof of convergence of Algorithm~\ref{algo:line_search} is given in Appendix~\ref{sec:proof_of_convg_line_search}.

\subsection{Threshold for Symmetric GAD and Depolarizing Channels}
\label{sec:symGADDepolConverseRegion}

\subsubsection{Symmetric GAD Channel}
\label{sec:symGADConverseRegion}

Let us compute converse regions when quantum states are corrupted by GADCs. We only consider symmetric GADC (with $\mu = 1/2$), since its classical capacity is additive; for $\mu \neq 1/2$, the additivity of classical capacity is not known. Substituting classical capacity of symmetric GADC from \eqref{eqn:capacity_gadc} in the necessary condition for $\epsilon$-accuracy in \eqref{eqn:redundancy_bound_k} yields:
\begin{equation}
    g_{\bm{\theta}}(k, p_0) = \frac{c(\epsilon, \eta, n)}{k} - 1 + h_2\left( \frac{1-\sqrt{1-p(k)}}{2} \right) \leq 0.
    \label{eqn:gadc_g}
\end{equation}

In \eqref{eqn:gadc_g}, the last term is monotonic (increasing) in $p$, and $p(k; p_0, \bm{\theta})$ is monotonic (increasing) in $\bm{\theta}$. Therefore, Corollary~\ref{col:converse_region} also holds for symmetric GAD channel. Therefore, for a given $\bm{\theta}$, the threshold $p_{th}(\bm{\theta}, \epsilon)$ can be computed by solving bi-level optimization problem \eqref{eqn:bilevel}. However, we cannot obtain closed-form expressions like for the erasure channel due to the challenge from the binary entropy term in \eqref{eqn:gadc_g}; therefore, the threshold $p_{th}(\bm{\theta}, \epsilon)$ must be computed numerically.
Since, symmetric GAD channel is additive, and scale-dependent noise $p(k; p_0, \bm{\theta})$ is monotonic in $\bm{\theta}$ component-wise, the threshold $p_{th}(\bm{\theta}, \epsilon)$ can be computed using Algorithm~\ref{algo:bisec_gd} (Theorem~\ref{thm:algo1_converge} holds). 

However, since Holevo information of symmetric GADC is concave in $p$, even if $p(k; p_0, \bm{\theta})$ is convex in $k$, unlike the erasure case, $g_{\bm{\theta}}(k, p_0)$ is not convex in $k$. Therefore, the lower-level problem $P_L$ must be solved numerically using Algorithm~\ref{algo:line_search} to obtain the threshold $p_{th}(\bm{\theta}, \epsilon)$ for a given $\bm{\theta}$. For a polynomial noise model described in Section~\ref{sec:scale}, we can compute Lipschitz constant $L$ in closed form for a given $\bm{\theta}$ (Refer to Appendix~\ref{sec:lipschitz_const_gadc} for the derivation).

\subsubsection{Depolarizing Channel}
\label{sec:depolarizingConverseRegion}

In this section, we compute the converse region when computational states are corrupted by depolarizing noise. 
Substituting for the classical capacity of the depolarizing channel from \eqref{eqn:capacity_depolarizing} in \eqref{eqn:redundancy_bound_k}, we obtain
\begin{equation}
    g_{\bm{\theta}}(k, p_0) =  h_2\left(\frac{p(k)}{2} \right) -1 + \frac{c(\epsilon, \eta, n)}{k} \leq 0.
    \label{eqn:depol_g}
\end{equation}

Similar to the symmetric GAD channel, the first term is increasing in $p$, and $p(k; p_0, \bm{\theta})$ is non-decreasing in $\bm{\theta}$. Therefore, Corollary~\ref{col:converse_region} and computation of threshold $p_{th}(\bm{\theta}, \epsilon)$ by solving bi-level optimization problem \eqref{eqn:bilevel} also hold. Also, similar to symmetric GADC, since obtaining closed-form expressions for $p_{th}(\bm{\theta}, \epsilon)$ is not possible, it can be computed using Algorithm~\ref{algo:bisec_gd}. Since $g_{\bm{\theta}}(k, p_0)$ is non-convex (due to $h_2(\cdot)$ in \eqref{eqn:depol_g} being concave), the threshold $p_{th}(\bm{\theta}, \epsilon)$ can be computed using line-search (Algorithm~\ref{algo:line_search}). Again, similar to the symmetric GAD channel, Lipschitz constant $L$ can be computed in closed form for a given $\bm{\theta}$ (Refer to Appendix~\ref{sec:lipschitz_const_depol} for the derivation).

\section{Converse Region}

\label{sec:converse_region}
Beyond the notion of a scalar threshold $p_{th}(\bm{\theta}, \epsilon)$, there exists a more general notion of the converse region, which extends the concept of threshold to multiple parameters simultaneously. We specifically aim to characterize the set of $(p_0,\bm{\theta})$ for which $\epsilon$-accurate computation is not possible. The following corollary to Theorem \ref{thm:finiteLB} provides a converse in terms of $(p_0, \bm{\theta})$.

\begin{corollary}
\label{col:converse_region}
Suppose we have,
\begin{equation}
    \Bar{\Theta} \triangleq \left\{(p_0,\bm{\theta})  \Bigr| \underset{k \geq 1}{\min} \text{   } g_{\bm{\theta}}(k, p_0) \geq 0 \right\},
\end{equation}
where
\begin{equation}
    g_{\bm{\theta}}(k, p_0) \triangleq \frac{c(\epsilon, R_f, n)}{k} - \chi(\mathcal{N}_{p(k)}).
\end{equation}
Then $\epsilon$-accurate computation is not possible for $(p_0, \bm{\theta}) \in \Bar{\Theta}$. Also, if $(p_0,\bm{\theta}) \in \Bar{\Theta}$ then $(p'_0,\bm{\theta}') \in \Bar{\Theta}$ if $(p'_0,\bm{\theta}') \ge (p_0,\bm{\theta})$ in a component-wise sense.
\end{corollary}
\begin{proof}
From Definition~\ref{defn:epsaccuracy}, we must prove that if $P_e < \epsilon$, then $(p_0,\bm{\theta}) \notin \Bar{\Theta}$. From Theorem~\ref{thm:finiteLB}, we have if $P_e < \epsilon$, then
\begin{equation}
\begin{split}
    k &> \frac{c(\epsilon, R_f, n)}{\chi(\mathcal{N}_{p(k)})}, \\
    g_{\bm{\theta}}(k, p_0) &= \frac{c(\epsilon, R_f, n)}{k} - \chi(\mathcal{N}_{p(k)}) < 0. \label{eqn:g_ineq}
\end{split}
\end{equation}
For any $\bm{\theta}$, \eqref{eqn:g_ineq} is satisfied only if
\begin{equation}
    \min_{k \geq 1} \, g_{\bm{\theta}}(k, p_0) < 0.
\end{equation}
In other words, $(p_0,\bm{\theta}) \notin \Bar{\Theta}$.

As $\chi(\mathcal{N}_p)$ is non-increasing in $p$ and $p(k;p_0,\bm{\theta})$ is non-decreasing in each component, $(p'_0,\bm{\theta}') \ge (p_0,\bm{\theta})$ in a component-wise sense implies $(p'_0,\bm{\theta}') \in \Bar{\Theta}$ whenever $(p_0,\bm{\theta}) \in \Bar{\Theta}$.
\end{proof}

We refer to $\Bar{\Theta}$ as the \textit{converse region} since $\epsilon$-accurate classical computation  on quantum circuits is not possible if the parameters of the scale-dependent noise are in $\Bar{\Theta}$. As any fault-tolerant implementation has to avoid this region, characterizing $\Bar{\Theta}$ is of particular interest. By Corollary~\ref{col:converse_region}, for characterizing $\Bar{\Theta}$, it is enough to find the minimum $p_0$ for each $\bm{\theta}$ such that $(p_0,\bm{\theta}) \in \Bar{\Theta}$.

\begin{figure}
    \centering
     \includegraphics[scale=0.5]{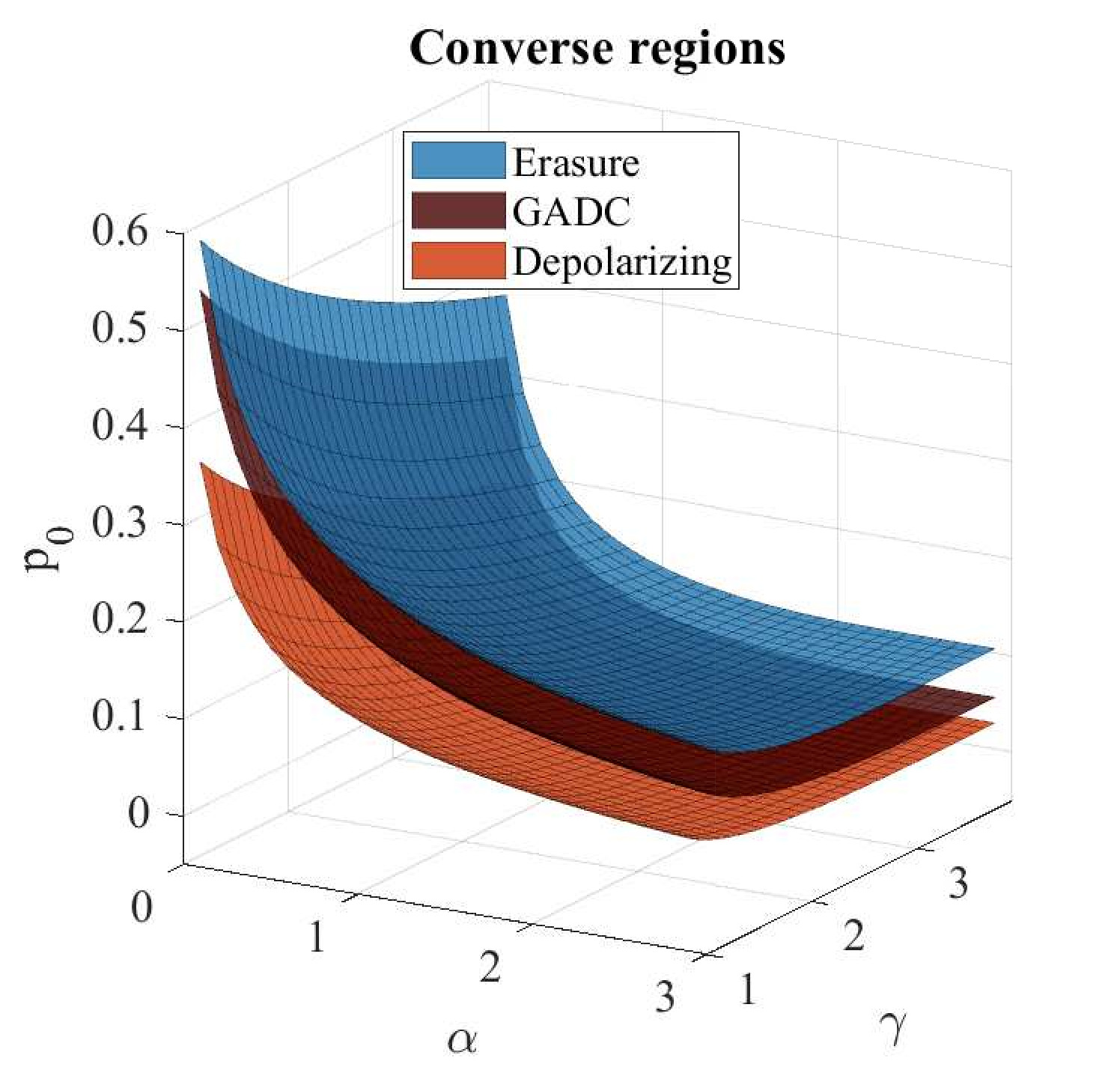}
     \caption{Comparison of converse regions (surfaces) for erasure, depolarizing and symmetric GADC with $\epsilon=0.1$, $\log_2R_f = n = 128$. The probability of error per physical qubit is assumed to scale with redundancy $k$ as $p(k;p_0, \bm{\theta}) = p_0 (1+\alpha (k-1))^{\gamma}$.}
    \label{fig:erasure_depol_gadc_surf_plot}
\end{figure}

Fig.~\ref{fig:erasure_depol_gadc_surf_plot} shows the converse regions when quantum computation is affected by erasure, depolarizing, and generalized amplitude damping noise. These are obtained by solving \eqref{eqn:bilevel} for $\epsilon=0.1$ using the aforementioned provably accurate algorithms. For a given $\bm{\theta} = (\alpha, \gamma)$ the thresholds are related as $p_{th}(\bm{\theta}, \epsilon)^{(e)} \geq p_{th}(\bm{\theta}, \epsilon)^{(g)} \geq p_{th}(\bm{\theta}, \epsilon)^{(d)}$ (point-wise), where the superscripts stand for erasure, symmetric GAD, and depolarizing channels, respectively. This relation is expected since Holevo information of the channels are related for a given $p \in (0,1)$ as $\chi^{(e)}(\mathcal{N}_p) \geq \chi^{(g)}(\mathcal{N}_p) \geq \chi^{(d)}(\mathcal{N}_p)$ (point-wise).

{The converse region only indicates that achieving $\epsilon$-accurate computation is not possible for the scaling parameters within that specific region. However, it does not imply that $\epsilon$-accurate computation is automatically feasible outside this region. It is important to note that using the bound only shrinks the size of the converse region compared to the actual impossibility region. The actual impossibility region could potentially be larger and also subsume the converse region.}

\section{Conclusion}
\label{sec:conclusion}
We considered a model of quantum circuits where inputs and outputs are classical, which includes a large class of algorithms due to Deutsch-Jozsa, Grover, and Shor. Using the example of the phase estimation circuit, we demonstrated that the currently best-known redundancy lower bounds for quantum computation are not applicable for quantum circuits with classical input and output. Then, we considered the scenario where quantum states are corrupted by i.i.d.\ (Holevo) additive quantum noise. We established a connection between the problem of noisy computation and noisy classical communication over a quantum channel and obtained a non-asymptotic lower bound on redundancy. 

Using this bound we studied fault-tolerant quantum computation under scale-dependent noise, where noise increases with added redundancy due to sharing of the limited physical resource. This led to two interesting thresholds on the original noise   $p_0$, beyond which, (i) redundancy impacts only adversely and (ii) no fault-tolerant scheme can achieve $\epsilon$-accuracy for a given $\epsilon$, respectively. The first threshold is similar to the one studied in \cite{fellous2021limitations} for concatenated codes, however, it is applicable to all fault-tolerant schemes. For the second threshold, we derived closed-form expressions whenever possible, and for other cases, we used optimization techniques for numerical characterizations.

{
In our derivation of the redundancy lower bound, certain relaxations were introduced to the noisy computation model. As a result, a potential limitation of our approach may be that the bound is not tight. In this paper, we have not explored the tightness of the bound, leaving it as a potential area for future research extensions. Future work could focus on either tightening the redundancy lower bound or establishing an achievability bound (redundancy upper bound) and characterizing the gap between the two, and its implications on the noise threshold.
}

\appendix

\section{Holevo Capacities of Erasure, Depolarizing, and GAD Channels}
\label{apndx:cptp_holevo_channels}
Here, we consider only Holevo-additive channels characterized by a single parameter $p \in [0,1]$ and whose Holevo capacity is monotonically decreasing in $p$. The candidate channels that we specifically study are (i) $p$-erasure, (ii) $p$-depolarizing and (iii) symmetric generalized amplitude damping channels, i.e., GADC$(p,\frac{1}{2})$.

\paragraph{Erasure Channel} In a quantum erasure channel (QEC), each qubit flips to $\ketbra{e}$, which is orthogonal to every $\rho \in L(\mathbb{C}^d)$, with probability $p$. Therefore, whenever a qubit gets corrupted, the location of corruption is known.
\begin{equation}
    \mathcal{N}_p(\rho) = (1-p) \rho + {p} \text{Tr}[\rho] \ketbra{e}.
\end{equation}

The classical capacity is \cite{khatri2020principles}:
\begin{equation}
    \chi(\mathcal{N}_p) = 1-p.
    \label{eqn:capacity_erasure}
\end{equation}

\paragraph{Depolarizing Channel}
When a qubit undergoes depolarizing noise, it is replaced by a maximally mixed state $I/2$ with probability $p$ \cite{khatri2020principles}: 
\begin{equation}
    \mathcal{N}_p(\rho) = (1-p) \rho + \tfrac{p}{2} I.
\end{equation}
In contrast to the erasure channel, the receiver (or the decoder) is not aware of the location of the error. The Holevo information of the depolarizing channel is:
\begin{equation}
    \chi(\mathcal{N}_p) = 1-h_2\left(\tfrac{p}{2}\right),
    \label{eqn:capacity_depolarizing}
\end{equation}
where $h_2(\cdot)$ is the binary entropy function. Note that the Holevo information is similar to the capacity of a binary symmetric channel with crossover probability $p/2$.

\paragraph{Generalized Amplitude Damping Channel (GADC)} Amplitude damping channels model the transformation of an excited atom to ground state by spontaneous emission of photons. The changes are expressed using $\ket{0}$ for the ground (no photon) state and $\ket{1}$ for the excited state. If the initial state of the environment $\ketbra{0}$, is replaced by the state $\theta_\mu \triangleq (1-\mu) \ketbra{0} + \mu \ketbra{1}, \mu \in [0,1]$ where, $\mu$ is thermal noise, we get the generalized ADC described using the following four Kraus operators \cite{khatri2020principles}:
\begin{equation}
\begin{split}
    A_1 &= \sqrt{1-\mu} \begin{bmatrix}
    1 & 0\\ 0 & \sqrt{1-p}
    \end{bmatrix}, \quad
    A_2 = \sqrt{1-\mu} \begin{bmatrix}
    0 & \sqrt{p}\\ 0 & 0
    \end{bmatrix}, \nonumber \\
     A_3 &= \sqrt{\mu} \begin{bmatrix}
    \sqrt{1-p} & 0\\ 0 & 1
    \end{bmatrix},  
    \quad A_4 = \sqrt{\mu} \begin{bmatrix}
    0 & 0\\ \sqrt{p} & 0
    \end{bmatrix}. 
\end{split}
\end{equation}
GADC is not additive in general (for arbitrary $\mu$). However, in the special case of symmetric generalized amplitude damping, i.e., generalized amplitude damping with  $\mu=1/2$, it is a Holevo additive channel. The classical capacity of symmetric GADC ($\mu=1/2$) is \cite{khatri2020information}:
\begin{equation}
    \chi(\mathcal{N}_p) = 1-h_2\left( \tfrac{1-\sqrt{1-p}}{2} \right),
    \label{eqn:capacity_gadc}
\end{equation}
where $p$ is the probability an atom decays from excited to ground state.

\begin{remark}
Note that we have used $p$ to describe different impairments in different channels, so  $p$ must be interpreted appropriately based on context.
\end{remark}

\section{Proof of Theorem~\ref{col:closed_erasure}}

\label{sec:erasure_closed_form}

The Holevo capacity of erasure channel is $\chi(\mathcal{N}_{p(k)})=1-p(k)$. Therefore, 
\begin{equation}
    g_{\bm{\theta}}(k,p_0) = \frac{c(\epsilon,R_f,n)}{k} + p(k) - 1. \label{Eq:g}
\end{equation}
If $\bm{\theta} \in \bar{\Theta}$ (converse region), then from Theorem~\eqref{col:converse_region} the following holds:

\begin{equation}
    g_{\bm{\theta}}(k, p_0) = \frac{c(\epsilon,R_f,n)}{k} + p(k) - 1 \geq 0, \quad \forall k \geq 1. \label{Eq:g}
\end{equation}
Differentiating w.r.t.\ $k$ and equate to 0 (to find stationary point),
\begin{equation}
         g'_{\bm{\theta}}(k, p_0) = -\frac{c(\epsilon,R_f,n)}{k^2} + p'(k;p_0, \bm{\theta}) = 0. \label{Eq:c}
\end{equation}
Henceforth, we shall use $c=c(\epsilon,R_f,n)$ for brevity. For a fixed $\bm{\theta} =  \alpha$ and $\gamma$ (respectively), the thresholds $p_{th}(\bm{\theta}, \epsilon)$ are derived for some well-behaved $p(k; p_0, \bm{\theta})$ as follows:
\begin{enumerate}
  
    \item  $p(k; p_0, \bm{\theta}) = p_0(1+\alpha (k-1))$: \\ 
    Suppose \eqref{Eq:g} holds for some $p_0$, then:
    \begin{equation}
    \begin{split}
        g_{\bm{\theta}}(k, p_0) &\geq 0, \quad \mbox{ for } k=1, \\
        c + p_0 -1 & \geq 0, \\
        p_0 &\geq 1-c. \label{eq:1linearcase}
    \end{split}
    \end{equation}
    Note that $k_{\max} = 1 + \alpha^{-1} ( (p_0)^{-1} - 1 )$ (obtained by solving for $k$ in $p(k;p_0, \bm{\theta}) = 1$). Also since $p(\cdot; p_0, \bm{\theta})$ is linear, $g_{\bm{\theta}}(., p_0)$ is convex in $[1,k_{\max}]$. Therefore, $g_{\bm{\theta}}(k, p_0)$ is minimized at any one of $k=1$, $k=k_{\max}$ or a stationary point in $(1, k_{\max})$. Substituting $p'(k;p_0, \bm{\theta}) = p_0 \alpha$ in \eqref{Eq:c}, the stationary point is:
    \begin{equation}
        k = \sqrt{\frac{c}{p_0 \alpha}}. \label{eq.linearroot}
    \end{equation}
          
        \begin{enumerate}[label=(\alph*)]
        \item Note that for $k \in (1, k_{\max})$ to be the minimum, $g'(k,p_0,\bm{\theta}, \epsilon)|_{k=1} < 0$. Also, noting that $p_0 \geq 1-c$ (from \eqref{eq:1linearcase}), we obtain
        \begin{equation}
            (1-c) \alpha \leq p_0 \alpha < c.
        \end{equation}
        Therefore, $\alpha < \frac{c}{1-c}$. Substituting \eqref{eq.linearroot} in \eqref{Eq:g}:
        \begin{equation}
        \begin{split}
            p_0 &\geq \frac{\left( \sqrt{c\alpha} - \sqrt{c\alpha - \alpha +1} \right)^2}{(\alpha -1)^2}, \nonumber \\
            p_{th}(\bm{\theta}, \epsilon) &= \frac{\left( \sqrt{c\alpha} - \sqrt{c\alpha - \alpha +1} \right)^2}{(\alpha -1)^2}. \nonumber
        \end{split}
        \end{equation}
        Note that since $\alpha < \frac{c}{1-c}$, the second term in the numerator, $c\alpha - \alpha +1 = 1-c \geq 0$. Therefore, the threshold $p_{th}(\bm{\theta}, \epsilon)$ exists.
        \item If $\alpha \geq \frac{c}{1-c}$, then
        \begin{align*}
            k &= \sqrt{\frac{c}{p_0 \alpha}} \leq \sqrt{\frac{1-c}{p_0}}; k \geq 1, \\
            p_0 &\leq 1 - c.
        \end{align*}
        However, $p_0 \geq 1-c$ from \eqref{eq:1linearcase}. Therefore, $p_{th}(\bm{\theta}, \epsilon) = 1-c$.
        \end{enumerate}
    \item $p(k;p_0, \bm{\theta}) = p_0k^{\gamma}$:
    \\ Here, $\bm{\theta} = \gamma$. The value of k ranges from $1 \leq k \leq \left(\frac{1}{p_0}\right)^{\frac{1}{\gamma}}$. For the choice of $p_{th}(\bm{\theta}, \epsilon)$, \eqref{Eq:g} must hold for all $k$ in this range. Similar to linear case, for this choice of $p(k)$ and range of $k$, \eqref{Eq:g} is convex. Hence for $p_{th}(\bm{\theta}, \epsilon)$, 
    \begin{equation}
    \begin{split}
        g_{\bm{\theta}}(k, p_0) \mid _{k=1} & \geq 0 \\
        c + p_0 -1 & \geq 0 \\
        p_0 & \geq 1-c. \label{eq:1expcase}
    \end{split}
    \end{equation}
    Substituting $p'(k;p_0, \bm{\theta}) = p_0\gamma k^{\gamma - 1}$ in \eqref{Eq:c}, we obtain the stationary point as: 
    \begin{equation}
        k = \left(\frac{c}{p_0\gamma}\right)^{\frac{1}{\gamma + 1}}.
        \label{eq:exproot}
    \end{equation}
    Similar to linear $p(k;p_0, \bm{\theta})$, there are two cases:
        \begin{enumerate}[label=(\alph*)]
            \item If $\gamma < \frac{c}{1-c}$,\\
            Substituting the stationary point computed in \eqref{eq:exproot} in \eqref{Eq:g}, the threshold $p_{th}(\bm{\theta}, \epsilon)$ can be computed as:
            \begin{align*}
                & p_0 \geq \frac{\left(\frac{\gamma}{c} \right)^{\gamma}}{\left(\gamma + 1 \right)^{\gamma + 1}}, \\
                & p_{th}(\bm{\theta}, \epsilon) = \frac{\left(\frac{\gamma}{c} \right)^{\gamma}}{\left(\gamma + 1 \right)^{\gamma + 1}}. 
            \end{align*}
            \item If $\gamma \geq \frac{c}{1-c}$, then,
            \begin{align*}
                k &= \left(\frac{c}{p_0\gamma}\right)^{\frac{1}{\gamma + 1}} \leq \left(\frac{1-c}{p_0}\right)^{\frac{1}{\gamma + 1}};~k \geq 1, \\
                p_0 &\leq 1-c.
            \end{align*}
            However, $p_0 \geq 1-c$ from \eqref{eq:1expcase}. Therefore, $p_{th}(\bm{\theta}, \epsilon) = 1-c$.
        \end{enumerate}
\end{enumerate}

\section{Proof of Convergence of Algorithm~\ref{algo:bisec_gd}}
\label{sec:proof_of_convg_bisec}

The following theorem provides a proof of global convergence of Algorithm~\ref{algo:bisec_gd}, with only a monotonicity assumption in $\bm{\theta}$ (note that continuity in $\bm{\theta}$ is not needed).

\begin{theorem}
Suppose a quantum circuit is corrupted by a scale-dependent noise-per-physical qubit, $p(k; \bm{\theta})$ that is monotonic in $\bm{\theta}$. Then for any given $\bm{\theta}$, the sequence $\{p_{0_i}\}$ generated using Algorithm~\ref{algo:bisec_gd} converges to the threshold $p_{th}(\bm{\theta}, \epsilon)$.
\label{thm:algo1_converge}
\end{theorem}
\begin{proof}
Algorithm \ref{algo:bisec_gd} generates a non-increasing sequence $\{p^+_{0_i}\}$ and a non-decreasing sequence $\{p^-_{0_i}\}$, which at every iteration yields $g^*(p^+_{0_i}) \geq 0$ and $g^*(p^-_{0_i}) < 0$, with $p_{0_i} = \tfrac{p^+_{0_i} + p^-_{0_i}}{2}$. Since the bisection method halves the difference between $p^+_{0_i}$ and $p^-_{0_i}$ at every iteration (i.e., $p^+_{0_{i+1}}-p^-_{0_{i+1}} = \frac{p^+_{0_{i}}-p^-_{0_{i}}}{2}$), we have that for all $\epsilon > 0$, there exists an $i_0$ such that for all $i \geq i_0$,  we get $p^+_{0_i} - p^-_{0_i} < \epsilon$. Also, since both $\{p^+_{0_i}\}$ and $\{p^-_{0_i}\}$ are bounded, they converge, and since for all $i \geq i_0$,  $p^+_{0_i} - p^-_{0_i} < \epsilon$, they converge to a common limit point (say $p^*$). Due to the monotonicity of $g^*_{\bm{\theta}}(p_0)$ (non-decreasing with $p_0$), the following inequality holds: $g^*_{\bm{\theta}}(p^-_{0_i}) \leq g^*_{\bm{\theta}}(p^*_0) \leq g^*_{\bm{\theta}}(p^+_{0_i})$. Therefore, $g^*_{\bm{\theta}}(p_0)<0$, for all $p_0 < p^*$, and $g^*_{\bm{\theta}}(p_0) \geq 0$, for all $p_0 > p^*$, which is by definition $p^* = p_{th}(\bm{\theta}, \epsilon)$.
\end{proof}

\section{Proof of Convergence of Algorithm~\ref{algo:proj_gd}}

\label{sec:proof_of_convg_proj_gd}

The following theorem provides proof of convergence of Algorithm~\ref{algo:proj_gd}. For better readability, the associated lemmas used in the proof are included in Section~\ref{apndx:lemmas_grad_desc}.

\begin{theorem}
\label{thm:proj_gd_converge}
\textit{Convergence of Algorithm~\ref{algo:proj_gd}:} Suppose $g^*_{\bm{\theta}}(p_0) = \underset{k \geq 1}{\min}~g_{\bm{\theta}}(k, p_0)$, which is convex in $k$. Then Algorithm~\ref{algo:proj_gd} yields $\tilde{g}$ arbitrarily close to $g^*_{\bm{\theta}}(p_0)$, i.e., for any pre-determined $\delta > 0$, $|\tilde{g} - g^*_{\bm{\theta}}(p_0)| \leq \delta$.
\end{theorem}
\begin{proof}
Let $\{1,\ldots,k_l\}$ be a sequence generated by projected gradient descent, \texttt{ProjGD}, where $k_l$ satisfies the stopping criterion. Note that \texttt{ProjGD} does not cross any stationary point if the step-size $\xi \leq \frac{1}{L}$ (from Lemma~\ref{lemma:gd_nocross_stationary}). So, $k_l = 1$ if and only if $\tilde{g} = g^*_{\bm{\theta}}(p_0) = g_{\bm{\theta}}(1, p_0)$, and similarly $k_l = k_{\max}$ if and only if $\tilde{g} = g^*_{\bm{\theta}}(p_0) = g_{\bm{\theta}}(k_{\max}, p_0)$. Otherwise $k_l \in (1, k_{\max})$ and $g_{\bm{\theta}}(1, p_0) < 0$, which implies from Lemma~\ref{lemma:gd_nocross_stationary} that $g_{\bm{\theta}}(k_l, p_0) \leq 0$. From Lemmas~\ref{lemma:stop_boundgrad} and \ref{lemma:converge_bounded_grad}, $k_l$ satisfying the stopping criterion in Def.~\ref{def:gd_stop_criterion1} is sufficient for convergence, i.e., $\tilde{g} = g_{\bm{\theta}}(k_l, p_0)$ and $|\tilde{g} - g^*_{\bm{\theta}}(p_0)| \leq \delta$.
\end{proof}

The following lemma shows $g'_{\bm{\theta}}(k,p_0)$ is indeed $L$-Lipschitz over $[1, k_{\max}]$ for a general polynomial noise model and gives a closed-form expression for $L$.

\begin{lemma}
\label{lemma:lipschitz_const_erasure}
\textit{Computing Lipschitz constant $L$}: $g'_{\bm{\theta}}(k,p_0)$ is $L$-Lipschitz over $[1, k_{\max}]$ for scale-dependent erasure noise $p(k;p_0, \bm{\theta}) = p_0 (1+\alpha (k-1))^\gamma$ with $\gamma \geq 0$ where, for $c=c(\epsilon,\eta,n)$,
\begin{equation}
    L = 2c + \alpha^2 \gamma |\gamma-1| p_0^{\min\left\{1, \frac{2}{\gamma}\right\}}.
    \label{eqn:lipschitz_const_erasure}
\end{equation}
\begin{proof}
Let $g'_{\bm{\theta}}$ and $g''_{\bm{\theta}}$ denote the partial derivatives $\frac{\partial}{\partial k} g_{\bm{\theta}}(k,p_0)$ and $\frac{\partial^2}{\partial k^2} g_{\bm{\theta}}(k,p_0)$, respectively. The magnitude of the second order partial derivative is bounded above as:
\begin{equation}
    |g''_{\bm{\theta}}| \leq \underset{k}{\sup} \left|\frac{2c}{k^3}\right| + \underset{k}{\sup} \left| p''(k;p_0, \bm{\theta}) \right|,
\end{equation}
where the inequality follows from triangle inequality and maximizing each summand.
Observe that the first summand is maximized when $k=1$, and the second term is bounded above as
\begin{equation}
\begin{split}
    p''(k;p_0, \bm{\theta}) & \leq 
    \begin{cases}
        \alpha^2 \gamma~|\gamma-1|~p_0, & 0 \leq \gamma < 2,\, k=1, \mbox{ and}\\
        \alpha^2 \gamma~|\gamma-1|~p_0^{2/\gamma}, & \gamma \geq 2,\, k = k_{\max},
    \end{cases} \\
    & \leq \alpha^2 \gamma |\gamma-1| p_0^{\min\left\{1, \frac{2}{\gamma}\right\}}, \quad\,\,\,\, \gamma \geq 0,
\end{split}
\end{equation}
where $k_{\max} = 1+\alpha^{-1} (p_0^{-(1/\gamma)}-1)$. Therefore,
\begin{equation}
    g''_{\bm{\theta}} \leq 2c + \alpha^2 \gamma~|\gamma-1|~ p_0^{\min\left\{1, \frac{2}{\gamma}\right\}} =: L.
\end{equation}
\end{proof}
\end{lemma}

\section{Proof of Convergence of Algorithm~\ref{algo:line_search}}

\label{sec:proof_of_convg_line_search}
Theorem~\ref{thm:line_search_converge} proves convergence of Algorithm~\ref{algo:line_search}. Required lemmas are in Section~\ref{apndx:lemmas_perturb}.

\begin{theorem}
\label{thm:line_search_converge}
\textit{Proof of convergence of Algorithm~\ref{algo:line_search}}: Algorithm~\ref{algo:line_search} yields $\tilde{g}$, which is arbitrarily close to $g^* = {\min}_{k \in[1, k_{\max}]} g_{\bm{\theta}}(k, p_0)$, i.e., $|\tilde{g} - g^*| \leq \delta$, for a pre-determined $\delta>0$.
\end{theorem}
\begin{proof}
Suppose $\{\ldots,k_i,k^-_i,k_{i+1},k^-_{i+1},\ldots\}$ is the sequence generated by Algorithm~\ref{algo:line_search}. From Lemma~\ref{lemma:gd_nocross_stationary} there are no stationary points in $(k_i, k^-_i)$. Then, the \texttt{Perturb} routine keeps track of the minimum value of $g_{\bm{\theta}}(\cdot, p_0)$ in $[k^-_i, k_{i+1}]$ at discrete increments: $\hat{g}_i = {\min}_{k \in \{k^-_i, k^-_i+\Delta k, \ldots, k_{i+1}\}} g_{\bm{\theta}}(k, p_0)$. This is followed by executing \texttt{ProjGD} again from $k_{i+1}$ to $k^{-}_{i+1}$, and so on. In every call to the \texttt{Perturb} routine, $\tilde{g}$ tracks the minimum of $\hat{g}_i$ until the $i$th iteration. From Lemma~\ref{lemma:max_perturb}, $\hat{g}_i$ differs from ${\min}_{k \in [k^-_i, k_{i+1}]} g_{\bm{\theta}}(k, p_0)$ by at most $\delta$. In line~\ref{algo:line_search_check_min_init} of Algorithm~\ref{algo:line_search}, $\tilde{g}$ is initialized with minimum at boundary points $k=\{1,k_{\max}\}$. Therefore, $\tilde{g} - g^* \leq \delta$. Finally, from Corollary~\ref{col:min_delk_nonconvex}, Lemma~\ref{lemma:max_perturb} and Lemma~\ref{lemma:finite_stationary}, Algorithm 3 terminates in finite steps when $k_{j} = k_{\max}$ or $k^-_{j} = k_{\max}$ for some $j \geq i+1$.
\end{proof}

\section{Lemmas for Convergence of Algorithms}

\subsection{Restriction of the feasible set of $P_L$ to $[1,k_{\max}]$}
\label{apndx:k_max}

Let $k_{\max} = \max\{k \mid p(k;p_0, \bm{\theta}) \leq 1\}$. If $k_{\max} = \infty$, then solving \eqref{eqn:bilevel} yields $p_{th}(\bm{\theta}, \epsilon) = 1$. Therefore, \eqref{eqn:bilevel} is non-trivial only if $k_{\max}$ is finite. Let $g_1(p_0) = {\min}_{k \geq 1}~g_{\bm{\theta}}(k, p_0)$ and $g_2(p_0) = {\min}_{k \in [1, k_{\max}]}~g_{\bm{\theta}}(k, p_0)$. From \eqref{eqn:erasure_g}, it can be observed that $g_1(p_0) = 0$ whenever $g_2(p_0) > 0$, and $g_1(p_0) = g_2(p_0)$ whenever $g_2(p_0) \leq 0$. Hence, the threshold $p_{th}(\bm{\theta}, \epsilon)$ obtained using $g_1(\cdot)$ and $g_2(\cdot)$ as a solution to $P_L$ in \eqref{eqn:bilevel} are identical. Therefore, \eqref{eqn:bilevel} can be equivalently solved by restricting the domain of $g_{\bm{\theta}}(\cdot, p_0)$ in $P_L$ to $[1, k_{\max}]$. In other words, one can replace line \ref{algo:PL} with $g^*_{\bm{\theta}}(p_0) = \min_{k \in [1,k_{\max}]} g_{\bm{\theta}}(k,p_0)$ to obtain the same value of threshold $p_{th}(\bm{\theta}, \epsilon)$. Additionally, this restriction makes the feasible set compact. Moreover, notice that the restriction and equivalence hold for all channels (not just erasure) as long as $\chi(\mathcal{N}_{p(k)}) = 0$ whenever $p(k) = 1$.

\subsection{Derivation of Lipschitz Constants for Symmetric GADC and Depolarizing channel}

\subsubsection{Symmetric GADC}
\label{sec:lipschitz_const_gadc}

\begin{lemma}
\label{lemma:lipschitz_const_gadc}
\textit{\textit{Computing Lipschitz constant $L$}}: $g'_{\bm{\theta}}(k,p_0)$ is $L$-Lipschitz over $[1, k_{\max}]$ for a polynomial scale-dependent symmetric GAD noise $p(k;p_0, \bm{\theta}) = p_0 (1+\alpha (k-1))^\gamma$, where
\begin{equation}
    \label{eqn:lipschitz_const_gadc}
    L := 2c + \frac{\alpha^2 \gamma}{2\ln{2}}\left( |\gamma-1| + \frac{\gamma}{3}\right).
\end{equation}

\end{lemma}
\begin{proof}

Denote $q(p) = q(p(k)) = \frac{1-\sqrt{1-p(k)}}{2}$, and $p = p(k;p_0, \bm{\theta})$ (for brevity); the magnitude of the second-order derivative of $g_{\bm{\theta}}(k, p_0)$ is bounded above as:

\begin{equation}
\begin{split}
    |g''_{\bm{\theta}}| &\leq \underset{k}{\sup} ~\left|\frac{2c}{k^3}\right| + \underset{k}{\sup}  ~|h''_2\left(q(k)\right)|, \nonumber \\
    &\leq 2c + \underset{k}{\sup}~\{ |h_2''(q) q'(p)^2 p'^2 \nonumber \\ 
    + & ~h_2'(q) q''(p) p'^2| \} + \underset{k}{\sup}~\{|h_2'(q) q'(p) p''|\}
\end{split}
\end{equation}
The last two terms on the right-hand side of the inequality are maximized when $k=k_{\max}$. The second term and the third term are bounded above as:
\begin{equation}
\begin{split}
    |&h_2''(q)q'(p)^2 p'^2 + h_2'(q) q''(p) p'^2| \\ \nonumber
    & \leq \alpha^2 \gamma^2 \left| \frac{p/(1-p)}{4 \ln 2} \ln\left( \frac{1}{e} \left(\frac{1-q(p)}{q(p)}\right)^{\frac{p}{2\sqrt{1-p}}}\right) \right| \leq \frac{\alpha^2 \gamma^2}{6 \ln 2},
\end{split}
\end{equation}
and
\begin{equation}
    h_2'(q) q'(p) p'' \leq \frac{\alpha^2 \gamma |\gamma-1|}{4 \sqrt{1-p}} \log_2\left(\frac{1-q(p)}{q(p)}\right) \leq \frac{\alpha^2 \gamma |\gamma-1|}{2 \ln{2}},
\end{equation}
respectively. The last inequalities in both the terms are obtained by allowing $p \rightarrow 1$ (i.e., $k \rightarrow k_{max}$). Therefore,
\begin{equation}   
\begin{split}
    |g''_{\bm{\theta}}| \leq 2c + \frac{\alpha^2 \gamma}{2 \ln{2}}\left( |\gamma-1| + \frac{\gamma}{3 }\right) =: L.
\end{split}
\end{equation}

\end{proof}

\subsubsection{Depolarizing Channel}
\label{sec:lipschitz_const_depol}

\begin{lemma}
\label{lemma:lipschitz_const_depol}
{\textit{Computing Lipschitz constant $L$}: $g'_{\bm{\theta}}(k,p_0)$ is $L$-Lipschitz over $[1, k_{\max}]$ for a polynomial scale-dependent depolarizing noise $p(k;p_0, \bm{\theta}) = p_0 (1+\alpha (k-1))^\gamma$, where}
\begin{equation}
    \label{eqn:lipschitz_const_depol}
    L = 2c + \frac{\alpha^2 \gamma}{2 \ln 2} \left( \frac{2 \gamma}{p_0 (2-p_0)} + |\gamma - 1| \ln\left( \frac{2-p_0}{p_0} \right) \right).
\end{equation}

\end{lemma}
\begin{proof}

The second order derivative of $g_{\bm{\theta}}(k, p_0)$ is bounded above as:
\begin{equation}
\begin{split}
    |g''_{\bm{\theta}}| &\leq \underset{k}{\sup} \left|\frac{2c}{k^3}\right| + \underset{k}{\sup} \left| h_2''(p(k)/2) \right|, \\
    |g''_{\bm{\theta}}| &\leq 2c + \underset{k}{\sup} \left\{\frac{1}{4} |h_2''(z)| p'(k;p_0, \bm{\theta})^2 + \frac{1}{2} h'(z) p''(k;p_0, \bm{\theta}) \right\},
\end{split}
\end{equation}
where $z = p(k;p_0, \bm{\theta})/2$. Noting that the $|h_2''(z)|$ and $h'(z)$ are maximized when $p(k;p_0, \bm{\theta}) = p_0$, we obtain

\begin{equation}
\begin{split}
    |g''_{\bm{\theta}}| \leq 2c + \frac{\alpha^2 \gamma}{2 \ln 2} & \left( \frac{2 \gamma }{p_0 (2-p_0)} \right. \\
    & \left. + ~ |\gamma-1|~\ln \left(\frac{2-p_0}{p_0}\right)  \right) =: L.
\end{split}
\end{equation}
\end{proof}

\subsection{Lemmas: projected gradient descent}

\label{apndx:lemmas_grad_desc}

\begin{definition}
\label{def:g_4_lemmas}
In Appendices~\ref{apndx:lemmas_grad_desc} and \ref{apndx:lemmas_perturb}, we consider $g(\cdot)$ to be of the following form: $g:[1,k_{\max}] \rightarrow \mathbbm{R}:k \mapsto g(k)$, where $g'(\cdot)$ is $L$-Lipschitz.
\end{definition}

\begin{lemma}
\label{lemma:stop_boundgrad}
\textit{Stopping criterion and bounded gradient:} Suppose a pair of iterates $(k_j, k_{j+1})$, which lie in the interior $(1,k_{\max})$, generated by \texttt{ProjGD} satisfy the stopping criterion  $|g(k_j) - g(k_{j+1})|< \frac{\delta^2}{2L k^2_{\max}}$, then the first order derivative is bounded above as $|g'(k_j)| < \frac{\delta}{k_{\max}}$.
\end{lemma}
\begin{proof}
    Applying the descent lemma to $k_j, k_{j+1}$, we get
\begin{equation}
\begin{split}
    g(k_{j+1},p_0) \leq g(k_{j},p_0) & \\
    +  g'(k_j,p_0) (k_{j+1} & - k_j) + \frac{1}{2} L |k_{j+1} - k_j|^2. \label{eqn1:convg_bound_grad} 
\end{split}
\end{equation}
Substituting  $k_{j+1} - k_j = - \xi g'_{\bm{\theta}}(k_j)$ in \eqref{eqn1:convg_bound_grad}:
\begin{equation}
    \xi \left(1 - \frac{\xi L}{2}\right) |g'(k_j)|^2 \leq g(k_{j}) - g(k_{j+1}).
\end{equation}
Choosing $\xi = \frac{1}{L}$, we obtain:
\begin{equation}
    \frac{1}{2L}|g'(k_j)|^2 \leq g(k_j) - g(k_{j+1}) < \frac{\delta^2}{2 L k^2_{\max}}.
\end{equation}
Therefore,
\begin{equation}
    |g'(k_j)| < \frac{\delta}{k_{\max}}.
\end{equation}
\end{proof}

\begin{lemma}
\label{lemma:converge_bounded_grad}
 Suppose $g(\cdot)$ is convex, and $g^* = {\min}_{k \in [1, k_{\max}]}~g(k)$. If $|g'(k_j)| < \delta/k_{\max}$, then $|g(k_j) - g^*| \leq \delta$, for any $k_j \in [1,k_{\max}]$.
\end{lemma}
\begin{proof}
From the convexity of $g(\cdot)$, we have $g(k) \geq g'(k_j) (k - k_j) + g(k_j)$, for any $k, k_j \in [1,k_{\max}]$. If $g'(k_j) < 0$, then:
\begin{equation}
    g(k) - g(k_j) \geq g'(k_j) (k_{\max} - k_j) \geq g'(k_j) (k_{\max} - 1), \forall k. \nonumber
\end{equation}
Therefore,
\begin{equation}
    g(k_j) - g^* \leq |g'(k_j)| (k_{\max} - 1) \leq |g'(k_j)| k_{\max} \leq \delta.
\end{equation}
On the other hand, if $g'(k_j) \geq 0$, then
\begin{equation}
    g(k) - g(k_j) \geq g'(k_j) (k - k_j) \geq g'(k_j) (1 - k_j), \forall k. \nonumber
\end{equation}
\begin{equation}
    g(k_j) - g^* \leq g'(k_j) (k_{\max} - 1) \leq g'(k_j) k_{\max} \leq \delta.
\end{equation}
Therefore, combining both cases: if $|g'(k_j)| \leq \frac{\delta}{k_{\max}}$, then $|g(k_j) - g^*| \leq \delta$.
\end{proof}

\begin{lemma}
\label{lemma:gd_nocross_stationary}
\textit{Projected gradient descent (\texttt{ProjGD}) does not cross any stationary point:} Let $k_j$ and $k_{j+1}$ be the successive iterates generated by \texttt{ProjGD} routine for $g(\cdot)$. Suppose, the step-size $\xi \in (0, \frac{1}{L}]$, then $g'(k_j) g'(k_{j+1}) \geq 0$.
\end{lemma}
\begin{proof}
    From the definition of Lipschitz gradient, we have $|g'(k_j) - g'(k_{j+1})| \leq L |k_j - k_{j+1}| = L \xi |g'(k_j)|$, where the last equality holds, since $k_{j+1}$ is generated from \texttt{ProjGD} routine. Suppose, $g'(k_j) \geq 0$, then the following inequalities hold:
    \begin{align*}
        -L\xi g'(k_j) &\leq g'(k_j) - g'(k_{j+1}) \leq L \xi g'(k_j),   \\
        -L\xi g'(k_j) &\leq g'(k_{j+1}) - g'(k_j) \leq L \xi g'(k_j), \\        
        (1-L\xi) g'(k_j) &\leq g'(k_{j+1}) \leq (1 + L \xi) g'(k_j).
    \end{align*}
    For $\xi \leq \frac{1}{L}$, we obtain:
    \begin{equation}
        g'(k_{j+1}) \geq g'(k_j) (1 - L \xi) \geq 0.
    \end{equation}
    Symmetrically, if $g'(k_{j}) \leq 0$, then $g'(k_{j+1}) \leq 0$. Combining both cases, we obtain $g'(k_j) g'(k_{j+1}) \geq 0$.
\end{proof}

\begin{lemma}
\label{lemma:projgd_min_delk}
\textit{Least difference between the successive iterates of \texttt{ProjGD}}: Let $k_j$ and $k_{j+1}$ be the successive iterates generated by \texttt{ProjGD} routine for $g(\cdot)$, with a step size $\xi = \frac{1}{L}$. If $|g(k_j) - g(k_{j+1})| > \zeta$, then $|k_j - k_{j+1}| > \sqrt{\frac{ 2\zeta}{5 L}}$.
\begin{proof}
    From the definition of Lipschitz gradient, we have
    \begin{equation}
        \begin{split}
            |g'(k_{j+1})-g'(k_{j})| &\leq L |k_{j+1} - k_j|, \\
            |g'(k_{j+1})| &\leq |g'(k_{j})| + L |k_{j+1} - k_j| \\
            & = \frac{|k_{j+1} - k_j|}{\xi} + L |k_{j+1} - k_j|.
        \end{split}
    \end{equation}
    Substituting $\xi = \frac{1}{L}$ we obtain
    \begin{equation}
        |g'(k_{j+1})| \leq 2L |k_{j} - k_{j+1}|. \label{eqn:bounded_derivative_kjplus1}
    \end{equation}
    Suppose $g(k_j) \geq g(k_{j+1})$, then using descent lemma \cite{bertsekas1997nonlinear} on $g(\cdot)$ at $k_j$ and $k_{j+1}$ yields
    \begin{align}
        g(k_j) - &g(k_{j+1}) \leq g'(k_{j+1}) (k_j - k_{j+1}) + \frac{L}{2} |k_j - k_{j+1}|^2 \label{eqn:descent_lemma1}\\
        & \leq |g'(k_{j+1})| |(k_j - k_{j+1})| + \frac{L}{2} |k_j - k_{j+1}|^2. \label{eqn:gkj_minus_gkjplus1_bound}
    \end{align}
    Substituting \eqref{eqn:bounded_derivative_kjplus1} in \eqref{eqn:gkj_minus_gkjplus1_bound}, and using $|g(k_j) - g(k_{j+1})| > \zeta$ we obtain
    \begin{equation}
        \zeta < |g(k_j) - g(k_{j+1})| \leq \frac{5 L}{2} |k_j - k_{j+1}|^2,
        \label{eqn:delk_interior}
    \end{equation}
    It can be also verified that \eqref{eqn:delk_interior} holds when $g(k_j) < g(k_{j+1})$. Therefore,
    \begin{equation}
        |k_{j+1} - k_j| > \begin{cases}
            \sqrt{\frac{ 2\zeta}{5 L}}, &  \mbox{if } k_{j+1} < k_{\max}, \\
            \sqrt{\frac{2}{5 L} (g(k_j) - g(k_{\max}))}, & \mbox{if } k_{j+1} = k_{\max}.
        \end{cases}
    \end{equation}    
    The second case is mentioned separately since \eqref{eqn:delk_interior} may not hold when $k_{j+1} = k_{\max}$.
\end{proof}
\end{lemma}

\begin{corollary}
\label{col:min_delk_convex}
Suppose \texttt{ProjGD} is used with stopping criterion 1 (Definition~\ref{def:gd_stop_criterion1}). If $|k_j - k_{j+1}| \leq \frac{\delta}{\sqrt{5} L k_{\max}}$, then $|g(k_j) - g(k_{j+1})| \leq \frac{\delta^2}{2 L k^2_{\max}}$, when $k_{j+1} < k_{\max}$.
\end{corollary}
\begin{proof}
The result follows by substituting $\zeta = \frac{\delta^2}{2 L k^2_{\max}}$ in Lemma~\ref{lemma:projgd_min_delk}.
\end{proof}

\begin{corollary}
\label{col:min_delk_nonconvex}
Suppose \texttt{ProjGD} is used with stopping criterion 2 (Definition~\ref{def:gd_stop_criterion2}). If $|k_j - k_{j+1}| \leq \sqrt{\frac{ 2\delta}{5 L}}$, then $|g(k_j) - g(k_{j+1})| \leq \delta$, when $k_{j+1} < k_{\max}$.
\end{corollary}
\begin{proof}
The result follows by substituting $\zeta = \delta$ in Lemma~\ref{lemma:projgd_min_delk}.
\end{proof}

\begin{remark}
  Corollaries~\ref{col:min_delk_convex} and \ref{col:min_delk_nonconvex} imply that \texttt{ProjGD} terminates in a finite number of iterations depending on the stopping criteria.
\end{remark}

\subsection{Lemmas: Perturbation}
\label{apndx:lemmas_perturb}

\begin{lemma}
\label{lemma:max_perturb}
\textit{\texttt{Perturb} routine does not miss stationary points:}
Suppose $k_j$ meets the stopping criterion 2 (Definition~\ref{def:gd_stop_criterion2}). If the perturbation $\Delta k \leq
\sqrt{\frac{2\delta}{L}}$, then the \texttt{Perturb} routine does not miss any stationary points with an error greater than $\delta$.
\end{lemma}
\begin{proof}
Let $k_j, k_{j+1}$ be any two points in $[1, k_{\max}]$. Since $g'(\cdot)$ is $L$-Lipschitz, using descent lemma \cite{bertsekas1997nonlinear} on $g(\cdot)$ at $k_j$ and $k_{j+1}$, we obtain
    \begin{equation}
        \label{eqn:descent_lemma1}
        g(k_j) - g(k_{j+1}) \leq g'(k_{j+1}) (k_j - k_{j+1}) + \frac{L}{2} |k_j - k_{j+1}|^2.
    \end{equation}

Let $k_{j+1}$ be the closest stationary point to $k_j$, then:
\begin{equation}
    |g(k_j) - g(k_{j+1})| \leq \frac{L}{2} |k_j - k_{j+1}|^2.
\end{equation}
Therefore, the following condition is necessary for the \textit{stopping criterion 2}, i.e., $|g(k_j) - g(k_{j+1})| \geq \delta$ (Definition~\ref{def:gd_stop_criterion2}) to hold:
\begin{equation}
     \Delta k = |k_j - k_{j+1}| \geq \sqrt{\frac{2\delta}{L}}.
\end{equation}
\end{proof}

\begin{lemma}
\label{lemma:finite_stationary}
\textit{
Upper bound on the number of stationary points:} Consider a set of stationary points $\{k_s\}$ of $g(\cdot)$ in $[1,k_{\max}]$ such that for every $k_s$, the adjacent stationary point $k_{s+1}$, $|g(k_s)-g(k_{s+1})| \geq \delta$. The number of such stationary points is finite and bounded above as $\lceil k_{\max}L/\delta \rceil$.
\end{lemma}
\begin{proof}
From the proof of Lemma \ref{lemma:max_perturb}, it follows that if $|g(k_s)-g(k_{s+1})| \geq \delta$, then the stationary points are separated by at least $|k_s-k_{s+1}| \geq \sqrt{\frac{2 \delta}{L}}$. Therefore, the number of stationary points in $[1, k_{\max}]$ is at most $\lceil k_{\max} \sqrt{\frac{L}{2 \delta}} \rceil$.
\end{proof}

\bibliography{bibfile}

\end{document}